\documentclass[a4paper,english]{lipics}
\usepackage{amsmath,amssymb,amsfonts,amsthm,color,datetime,tikz,mathtools,verbatim, booktabs}
\usetikzlibrary{automata,arrows,matrix}

\DeclareMathOperator{\block}{\boxdot}

\newcommand{\N}{\mathbb{N}}

\newcommand{\Ahat}{\hat{A}}
\newcommand{\B}{\mathcal{B}}
\newcommand{\gammaU}{\gamma_u}
\newcommand{\gammaV}{\gamma_v}

\newcommand{\VMon}{\mathbf{V}}

\newcommand{\Reg}{\mathbf{Reg}}
\newcommand{\MorphV}{\text{Hom}(A^*,\textbf{V})}

\newcommand{\pSub}{\pi_2}

\newcommand{\contPos}{\pi_c}
\newcommand{\lengthPres}{\lambda}

\newcommand{\cond}[1]{(#1)}

\DeclareMathOperator{\modp}{\mathrm{mod}_p}

\newcommand{\w}{w}
\newcommand{\s}{s}
\newcommand{\p}{p}

\newcommand{\ACzero}{\mathbf{AC}^0}
\newcommand{\NCzero}{\mathbf{NC}^0}

\newcommand{\ACCzero}{\mathbf{ACC}^0}
\newcommand{\CCzero}{\mathbf{CC}^0}
\newcommand{\ACp}[1]{\mathbf{ACC}^0[#1]}

\newcommand{\Parb}{\mathcal P_{\mathrm arb}}
\newcommand{\Parbbf}{\mathbf P_{\mathrm arb}}
\newcommand{\BlockP}{\mathcal V \block \Parb}
\newcommand{\BlockPA}{(\mathcal V \block \Parb)_{A}}

\newcommand{\CircBase}{\mathbf{B}}
\newcommand{\LangVar}{\mathcal{V}}

\IfFileExists{paper.rev}
{
\input{paper.rev}
}
{

}

\title{Using Duality in Circuit Complexity}
\author{Silke Czarnetzki}
\author{Andreas Krebs}
\affil{Wilhelm-Schickard-Institut, Universit\"at T\"ubingen\\
  Sand 13, 72076 T\"ubingen, Germany\\
  \texttt{\{czarnetz,krebs\}@informatik.uni-tuebingen.de}
}

\Copyright{Silke Czarnetzki and Andreas Krebs}

\begin{document}

\maketitle

\begin{abstract}
We investigate in a method for proving separation results for abstract classes of languages. 
A well established method to characterize varieties of regular languages are identities. We use a recently established generalization of these identities to non-regular languages by Gehrke, Grigorieff, and Pin: so called equations, which are capable of describing arbitrary Boolean algebras of languages.

While the main concern of their result is the existence of these equations, we investigate in a general method that could allow to find equations for language classes in an inductive manner. 

Thereto we extend an important tool -- the block product or substitution principle -- known from logic and algebra, to non-regular language classes. Furthermore, we abstract this concept by defining it directly as an operation on (non-regular) language classes. We show that this principle can be used to obtain equations for certain circuit classes, given equations for the gate types.

Concretely, we demonstrate the applicability of this method by obtaining a description via equations for all languages recognized by circuit families that contain a constant number of (inner) gates, given a description of the gate types via equations.
\end{abstract}

\section{Introduction}
In Boolean circuit complexity, deriving lower bounds on circuit size and depth has up to now shown to generally be difficult. While there have been results proving lower bounds, we still lack methods that are applicable in general.
Algebraic methods have improved our understanding of circuit complexity.
Here we are especially interested in the constant depth circuit complexity classes $\ACzero,\CCzero,$ and $\ACCzero$ that have tight connections to algebra via programs. For instance the class $\ACzero$ is equal to the class of languages recognized by polynomial-length programs over finite aperiodic monoids \cite{DBLP:journals/jacm/BarringtonT88}. 
Using these connections allowed the usage of algebraic methods in circuit complexity \cite{DBLP:conf/icalp/BarringtonT87,DBLP:journals/jcss/BarringtonS95,DBLP:conf/latin/BarringtonS95,DBLP:journals/cc/McKenziePT91,DBLP:journals/cc/Therien94}. For an overview see the book of Straubing \cite{Straubing:1994}.

It is a well known method from algebra to characterize regular language classes by identities and has successfully been applied to describe varieties of regular languages stemming from various logic classes (see for example the book of Pin \cite{Pin12}).
Recently, Gehrke, Grigorieff and Pin generalized the approach to work with non-regular language classes \cite{GGP}.

While many concrete characterizations via identities or equations exist for classes of regular languages (see for example the book of Almeida \cite{Almeida94}), only few concrete examples are known for non-regular classes \cite{GKP14}. 
One of the main difficulties is, that these equations hold for all languages in a circuit class and not only the regular ones, for which we have other manageable descriptions.
Furthermore, the question arises how to achieve an abstract method to obtain equations for circuit classes, instead of calculating them concretely for each class.

As the result of \cite{GGP} shows the existence of equational descriptions for arbitrary Boolean algebras, circuit classes form suitable candidates. However, it is not clear how to obtain these equations ina  constructive way. 
The method presented in the paper allows us to obtain equations for more complex classes of circuits, starting with equations from simple classes.
In this paper, we would like to dare a first step towards an approach to derive circuit lower bounds for abstract classes of circuits. Even though the classes described are fairly simple and seperation results could be proven by using combinatorical arguments, this is the first effort made towards a procedure to compute equations for more general circuit classes.

In order to reach our goal we abstract another powerful technique: the block product or substitution principle \cite{TeTh07}.
The idea of computing the defining equations for a more complex variety constructed by some principle from simpler varieties has been successfully used in the regular case \cite{AlWe98, KrSt13}. While all these previous results rely on regular language classes, we extended it to work on non-regular classes by defining an operation purely on language classes, not relying on monoids or automata, which reflects a decomposition of the computation of the circuit.

As our main contribution, we show that in principle it is possible to systematically construct equations for the block product under certain restrictions.  
To demonstrate that our method can be applied, we concretely compute the equations for languages recognized by constant size circuit families, given equations that describe the gate types allowed in the circuit family.

\subsection*{Organization of the paper}

We organized this paper in a way that all the definitions are introduced along with our demonstration of how to compute the equations for constant size circuit families.

As a first step, circuit classes whose gates are defined by a variety of languages are introduced in Section \ref{sec:circuits}.
In Section \ref{sec:blockproduct}, we define an abstract version of the block product.  Then in Section \ref{sec:duality} we introduce basic definitions and results from Stone duality as far as needed to formulate the main theorem in Section \ref{sec:equations}. We conclude and give hints for further research in Section \ref{sec:conclusion}.
And finally two section to prove soundness and completeness of the equations provided in the main theorem (Sections \ref{sec:soundness} and \ref{sec:completeness}).

\newpage


\section{Preliminaries}\label{pre}

\subsection{Varieties}\label{apx:varieties}

The term variety stems from the denotation of solutions to algebraic equations in the field of algebraic geometry. There is indeed a connection between equational theories and varieties, that will be pointed out later. Furthermore, varieties of languages admit for some useful properties that resemble those of classes definable by logic.

We will shortly explain the connection between varieties of monoids and varieties of languages, as they are directly related and we will use the connection between them. For further details, we refer to the book of Pin \cite{PinBook}.

A \emph{variety of monoids} is a class of monoids $\VMon$ satisfying the following properties
\begin{enumerate}
\item if $M \in \VMon$ and $N$ is a submonoid of $M$, then $N \in \VMon$
\item if $M \in \VMon$ and $N$ is a quotient of $M$, then $N \in \VMon$
\item if $(M_i)_{i \in I}$ is a finite family of monoids, then $\prod_{i \in I} M_i \in \VMon$
\end{enumerate}

A \textit{Boolean algebra} $\B$ over a set $X$ is a collection of subsets $X$, that is closed under finite intersections, finite unions and complement.

A \emph{variety of languages} is a class of regular languages $\mathcal{V}$ such that
\begin{enumerate}
\item for each alphabet $A$, $\mathcal{V}_A$ is a Boolean algebra over $A^*$
\item for each monoid morphism $\varphi : A^* \rightarrow B^*$, the condition $X \in \mathcal{V}_A$ implies $\varphi^{-1}(X) \in \mathcal{V}_B$
\item if $L \in \mathcal{V}_A$ and $a \in A$, then $a^{-1}L := \{w \in A^* \mid aw \in L\} \in \mathcal{V}_A$ and $La^{-1} \in \mathcal{V}_A$
\end{enumerate}

Eilenberg proposed in \cite{Eilenberg:1976:ALM:540244} that there exists a direct connection between varieties of monoids and varieties of languages, known as Eilenberg's variety theorem.

\begin{theorem}
Every variety of finite monoids corresponds directly to a variety of languages and vice versa.
\end{theorem}

That is to say that there are correspondences $\VMon \rightarrow \mathcal{V}$ and $\mathcal{V} \rightarrow \VMon$ that define mutually inverse bijective correspondences between varieties of finite monoids and varieties of languages. For a better intuition, the correspondence $\VMon \rightarrow \mathcal{V}$ is given by the variety of all languages recognisable by monoids of $\VMon$ and $\mathcal{V} \rightarrow \VMon$ is the class of all syntactic monoids of languages of $\mathcal{V}$.

\subsection{Circuits over words}\label{apx:circuits}

\begin{definition}[Boolean function]
A $k$-ary Boolean function $f$ is a map from $\{0,1\}^k$ to $\{0,1\}$. We say a collection of Boolean functions $(f_i)_{i\in\N}$ is a family of Boolean functions if for each $i \in \N$, $f_i$ is an $i$-ary Boolean function.
\end{definition}

We say a Boolean function $f$ is symmetric, if $f(x_0,\ldots,x_{k-1}) = f(x_{\sigma(0)}, \ldots, x_{\sigma(k-1)}$ for each permutation $\sigma$ of the numbers $\{0,\ldots,k-1\}$. Respectively, a family of Boolean functions is called symmetric, if each $f_i$ is.

\begin{definition}[Base]
A base is a set containing Boolean functions and families of Boolean functions.
\end{definition}

\begin{definition}[Circuit]
Let $A$ be a finite alphabet, and $\CircBase$ be a base.
A circuit over the base $\CircBase$ for words in $A^n$ is an acyclic directed graph with a unique node with fan-out 0. Nodes with fan-in 0 are called inputs and are labeled by $1$,$0$,$x_i\in S$, where $S \subseteq A$ and $i=\{0,\dots,n-1\}$.
All other nodes are called gates. Among these the unique node with fan-out 0 is called the output gate. Each gate with fan-in $k$ is labeled by a $k$-ary Boolean function or by a family of Boolean functions of the base $\CircBase$.
\end{definition}

We say that \emph{size} of a circuit is the number of its gates. Recall that inputs were not counted as gates and thus do not affect the size of the circuit. The \emph{depth} of a circuit is the maximal length of a path from an input to its output gate.

The evaluation of a word $w\in A^n$ for a circuit is defined inductively, as usual. The value of the input $0$ is always $0$, and the value of $1$ is $1$. The value the input $x_i$ labeled by $a \in A$ is $1$ if $w_i=a$, and $0$ otherwise.
The value of a gate with $k$ predecessors is its label function $f$ (or respectively $f_k$, if $f$ is a family of Boolean functions) applied to the values of the predecessors. This is well defined as we consider only bases containing symmetric functions (or families).  We say a circuit accepts $w$ if its output gate evaluates to $1$ for the word $w$.  A circuit family recognizes a language $L\subseteq A^*$, if for each $w\in L$, the circuit $C_{|w|}$ accepts $w$.

There are many well known circuit classes of constant depth and polynomial size, some of them listed in Figure \ref{fig:gatetype} (right table).

We show that the constant size circuit families over some base $\CircBase$ can be transformed to a single layer of $\CircBase$ gates with an $\NCzero$ circuit on top.

\begin{definition}[Circuit family]
We say $(C_i)_{i\in\N}$ is a family of circuits over a base $\CircBase$ for words in $A^*$ if there is a finite base $\mathbf F \subseteq\CircBase$ such that $C_i$ is a circuit over the base $\mathbf F$ for words in $A^i$.
The size (resp. depth) of $(C_i)_{i\in\N}$ is a map that maps $i$ to the size (resp. depth) of $C_i$.
\end{definition}

Please note that, as opposed to other common definitions, our definition does not require the base to be finite, but we limit each circuit family to use only a finite number of elements from the base.
This allows to define circuit classes like $\ACCzero$ directly over an infinite base instead of defining $\ACCzero$ as a union of $\ACp p$ each defined over a finite base.

For example the circuit class known as $\NCzero$ consists of constant size circuit families over the base  $\{\land_2,\lor_2,\lnot_1\}$ which consists only of Boolean functions and hence the fan-in of gates over such a base is bounded. As we are interested in constant size circuits that should access all inputs we consider bases that contain families of Boolean functions. The circuit class $\ACzero$ consists of polynomial size, constant depth circuit families over the base $\{\land=(\land_i)_{i\in\N},\lor=(\lor_i)_{i\in\N},\lnot_1\}$, where the base contains both families of Boolean functions and Boolean functions.

\subsection{Topology and the free profinite monoid}\label{apx:topology}

The notion of limit points allows us for equational characterizations of classes of languages, where explicit equations between words are not sufficient. Here, topology provides us with suitable tools to obtain these characterizations, as it allows for the definition of limits.

Let $\Omega$ be a set. A \textit{topology} $\mathcal{T}$ on $\Omega$ is a set of subsets of $\Omega$ with the following properties
\begin{enumerate}
\item $\emptyset, \Omega \in \mathcal{T}$
\item if $A_i \in \Omega$ for all $i$ in some index set $I$, then $\bigcup_{i \in I} A_i \in \mathcal{T}$ \\ ($\mathcal{T}$ is closed under arbitrary unions)
\item if $A_i \in \Omega$ for $i = 1,\ldots,n$ and $n \in \N$, then $\bigcap_{i=1}^n A_i \in \mathcal{T}$ \\ ($\mathcal{T}$ is closed under finite unions)
\end{enumerate}

A set $\Omega$ together with a topology $\mathcal{T}$ is called a topological space and denoted by $(\Omega,\mathcal{T})$.
The elements of $\mathcal{T}$ are called \textit{open sets}, whereas the complements of open sets are called \textit{closed sets}. Sets that are both open and closed are called \textit{clopen}. Furthermore, the \textit{closure} of a set $A$, denoted by $\overline{A}$ is the smallest closed set, that contains $A$.

We say a topological space $(\Omega,\mathcal{T})$ is \emph{hausdorff}, if for each two points $x,y \in \Omega$ there exist two disjoint open sets $O_x,O_y \in \mathcal{T}$ such that $x \in O_x$ and $y \in O_y$.

The clopen sets are of special interest to us. As seen in the introduction, we want to be able to show, that a language does not belong to some class of languages, by showing that it has non-empty intersection with its complement in some new space, when we add certain points. The closed sets allow for such a property, namely they are the only sets $A \subseteq \Omega$ such that $\overline{A} \cap \overline{A^c} = \emptyset$ holds. Thus, a relation between classes of languages and clopen sets of topological spaces is interesting and we will see, that such a connection indeed exists under some premises.

For example, the \emph{free profinite monoid} is a topological space, for which the clopens are exactly the closures of regular languages. But in order to define it we need the definition of a \emph{metric}.

A \textit{metric} on a space $\Omega$ is a function $d: \Omega \times \Omega \rightarrow [0,\infty)$ satisfying for all $x,y,z \in \Omega$
\begin{enumerate}
\item $d(x,y) = 0 \Leftrightarrow x = y$ \hfill ($d$ is positive definite)
\item $d(x,y) = d(y,x)$ \hfill ($d$ is symmetric)
\item $d(x,z) \leq d(x,y) + d(y,z)$ \hfill ($d$ satisfies the triangle inequality)
\end{enumerate}

A metric can be seen as a function that assigns a distance to two points.

The topology induced by a set $\mathcal{E} \subseteq \mathcal{P}(\Omega)$ is the smallest topology that contains $\mathcal{E}$. Let $x \in \Omega$ and let $d$ be a metric. The set of open $\epsilon$-balls centred at $x$
$$B_{\epsilon}(x) = \{y \in \Omega \mid d(x,y) < \epsilon\}$$
then induces a topology on $\Omega$. A space together with a topology induced by a metric is called a \textit{metric space} and denoted by $(\Omega,d)$.

Furthermore, every metric space is hausdorff.

We say a sequence $(x_n)_{n \in \N}$ of points in $\Omega$ \emph{converges} to a point $x \in \Omega$ with respect to some topology, if for each open set $O$, that contains $x$, there exists an $N \in \N$ such that $x_n \in O$ for all $n \geq N$. On a set that is equipped with the trivial topology, i.e. only $\emptyset$ and $\Omega$ are open, every sequence converges to every point.

We say a topological space $(\Omega,\mathcal{T})$ is \emph{compact}, if every open cover of $\Omega$ has a finite subcover. Formally that means, that every arbitrary family of open subsets $(U_i)_{i \in I}$ for some index set I with
$$\Omega = \bigcup_{i \in I} U_i$$
has a subfamily $(U_j)_{j \in J}$ where $J \subseteq I$ is finite, that satisfies
$$\Omega = \bigcup_{j \in J} U_j.$$

The following theorem states a useful property of compact spaces

\begin{theorem}
Let $(\Omega,\mathcal{T})$ be a compact topological space. Then every sequence has a converging subsequence.
\end{theorem}

A sequence $(x_n)_{n \in \N}$ is called a \textit{Cauchy sequence} with respect to some metric $d$, if for each $\epsilon > 0$ there exists an $N \in \N$ such that for all $n,m \geq N$ we have $d(x_n,x_m) < \epsilon$. Every sequence, that converges to some point in a metric space is a Cauchy sequence, but not every Cauchy sequence needs necessarily converge. 

We say a metric space is \textit{complete}, if every Cauchy sequence converges.
For any metric space, it is possible to define the distance of two Cauchy sequences $x = (x_n)_{n \in \N}$ and $y = (y_n)_{n \in \N}$ by setting $d'(x,y) = \lim_{n \rightarrow \infty} d(x_n,y_n)$. By $[\Omega]$ denote the space of all sequences in $\Omega$. Then $d'$ is not a metric on $[\Omega]$, since $d(x,y) = 0$ does not imply $x = y$, but it defines a metric on the space $[\Omega] /_{\sim}$ where $x \sim y$ if and only if $d'(x,y) = 0$.

The space $([\Omega] /_{\sim}, d')$ is a complete metric space, also called the \textit{completion} of $\Omega$, often denoted by $\widehat{\Omega}$ and has $\Omega$ as a dense subspace, i.e. every element of $\widehat{\Omega}$ is the limit of a sequence in $\Omega$. The elements of $\widehat{\Omega}$ are equivalence classes of Cauchy sequences. Furthermore, $\Omega$ can be embedded into $\widehat{\Omega}$ by mapping an element $x$ to the class of Cauchy sequences, that are eventually constant and equal to $x$.
In a more apprehensible way, we can think of the completion as the set together with all its convergence points.

Together with these definitions, we can define the free profinite monoid over an alphabet $A$. We say that a finite monoid $M$ separates two words $x$ and $y$, if there is a morphism $\varphi: A \rightarrow M$, with $\varphi(x) \neq \varphi(y)$. The function $d: A^* \times A^* \rightarrow [0,\infty)$ defined by
$$d(x,y) = \{ 2^{-\left| M \right|} \mid M \text{ separates $x$ and $y$} \}$$
then is a metric on $A^*$ and thus induces a topology. The topology itself is not of great interest, since it is discrete, i.e. every singleton is open. The \textit{free profinite monoid} $\Ahat$ is the completion of $A^*$ with respect to $d$.

The free monoid $A^*$, as stated before, can be embedded into $\Ahat$ by mapping a word $w$ to the class of Cauchy sequences, that are eventually constant and equal to $w$.

The free profinite monoid entails useful properties. Firstly, it is a compact and hausdorff space and it holds information about the regular languages, as the following theorem states.

\begin{theorem}\cite{Almeida94}
A language $L \subseteq A^*$ is regular if and only if $\overline{L} \subseteq \Ahat$ is clopen.
\end{theorem}

We say a function $g : X \rightarrow Y$ between two topological space is \emph{continuous} if for every open set $U$ of $Y$, $g^{-1}(U)$ is open in $X$. This translates to the sentence, that a function is continuous if preimages of open sets are again open. Since the preimage is closed under complement and equivalent definition is that preimages of closed sets are closed.
If $X$ and $Y$ are countable, this is equivalent to saying that for each sequence $(x_n)_{n \in \N}$ with $\lim_{n \rightarrow \infty} x_n = x \in X$, the property $\lim_{n \rightarrow \infty} g(x_n) = g(x)$ holds.

We can now define the unique property of $\Ahat$. For each morphism $\varphi: A^* \rightarrow M$ into a finite monoid $M$, there exists a unique continuous extension $\hat{\varphi}: \Ahat \rightarrow M$, with $\hat{\varphi}(a) = \varphi(a)$ for each $a \in A$.


\subsection{Stone duality}\label{apx:duality}

Indeed, the space $\Ahat$ is just a special case of a more general scenario.

Let $\B$ be a boolean algebra over a set $X$. An \textit{ultrafilter} of $\B$ is a non-empty subset $\gamma$ of $\B$ that satisfies
\begin{enumerate}
\item $\emptyset \notin \gamma$,
\item if $L \in \gamma$ and $K \supseteq L$, then $K \in \gamma$, \hfill ($\gamma$ is closed under extension)
\item if $L,K \in \gamma$ then $K \cap L \in \gamma$, \hfill ($\gamma$ is closed under finite intersections)
\item for each $L \in \B$, either $L \in \gamma$ or $L^c \in \gamma$ \hfill (ultrafilter condition)
\end{enumerate}

In topology, ultrafilters are often used as a generalization of sequences to uncountable spaces. Terms, such as convergence, can be defined likewise, such that they agree with the definitions on sequences for countable spaces.
One of the important properties of ultrafilters is that they have to choose. For any partition of elements from the Boolean algebra they have to contain exactly one partition element. Roughly speaking this makes ultrafilter like the atomic properties that a language in the Boolean algebra can have.

\begin{lemma}\label{lem:partitionimappendix}
Let $\mathcal B$ be a Boolean algebra an $\mu \in \mathcal S(\mathcal B)$. Further, let $\mathcal L = \{L_0,\ldots,L_{n-1}\}$ be a finite partition of $A^*$ of languages from $\B$.
Then there exists an $i \in \{0,\ldots,n-1\}$ such that $L_i \in \mu$.
\end{lemma}
\begin{proof}
Since $\mathcal L$ is a partition of $A^*$, we have that
$$ \bigcup_{i = 0}^{n-1} L_i \in \mu .$$
Suppose that for each $i \in \{0,\ldots,n-1\}$ the condition $L_i \notin \mu$ holds. Since $\mu$ is an ultrafilter, this implies that for all $i \in \{0,\ldots,n-1\}$, $L_i^c \in \mu$.
Observe that
$$ L_j = \bigcap_{\substack{i = 0\\ i \neq j}}^{n-1} L_i^c $$
and thus $L_j \in \mu$, which is a contradiction to the assumption.
\end{proof}

Each Boolean algebra has an associated compact Hausdorff space $\mathcal{S}(\B)$, called its \textit{Stone space}, often also referred to as its dual Space, which points may be given as the ultrafilters of $\B$, where the topology is induced by the sets $\widehat{L} = \{\gamma \in \mathcal{S}(\B) \mid L \in \gamma \}$ for $L \in \B$.

The dual Space contains the underlying set $X$ as a dense subset. Its embedding is given by
$$x \mapsto \{L \in \mathcal{B} \mid x \in L\}.$$
It is not too hard to verify, that this set satisfies conditions $1.-4.$ and thus is an ultrafilter of $\B$.

The clopen sets of $\mathcal{S}(\B)$ then are exactly the sets $\widehat{L}$ where $L \in \B$, which are the closures of elements of $\B$ with respect to the topology on $\mathcal{S}(\B)$.

Pippenger gave a proof in his paper \cite{pip1997}, that the Stone space of the regular languages $\mathcal{S}(\Reg)$ is no other than the free profinite monoid $\Ahat$. This justifies that any profinite word may be seen as an ultrafilter of $\Reg$ and the correspondence will often be used in the following. An alternate proof, that clarifies the structure of this connection, is given below.

We say two spaces are \emph{homeomorphic} if there exists a bijection $f$ such that both $f$ and the inverse function of $f$ are continuous with respect to the topology the spaces are equipped with. Informally speaking, this says the the spaces are equipped with the same topology. 

\begin{theorem}\label{AhatSReg}
The space $\Ahat$ is homeomorphic to $\mathcal{S}(\Reg)$.
\end{theorem}

\begin{proof}
For a Cauchy sequence $u$ with respect to the introduced metric on $\Ahat$, define the set
$$F_u = \{L \in \Reg \mid \exists n_0 \in \N \quad \forall n \geq n_0 : u_n \in L\}.$$
To see that $F_u$ defines an ultrafilter of $\Reg$ observe that $F_u$ is a non-empty set closed under intersections and upsets. What is left to show is, that for any regular language $L$, the set $F_u$ contains either $L$ or its complement. For that let $L$ be a regular language and $h: A^* \rightarrow M$ a morphism into a finite monoid recognising $L$, such that $L = h^{-1}(K)$ for some $K \subseteq M$. Since $u$ is a Cauchy sequence, there exists an $N \in \N$ such that for all $n,m \geq N$ we have $d(u_n,u_m) < 2^{- \left| M \right|}$ and thus $h(u_n) = h(u_m)$. Since either $h(u_N) \in K$ or $h(u_N) \in K^c$ and $L^c = h^{-1}(K^c)$, this implies that either $u_n \in L$ or $u_n \in L^c$ for all $n \geq N$. Thus, $F_u$ is an ultrafilter of $\Reg$.

Furthermore, let $u$ and $v$ be two equivalent Cauchy sequences. Then, by the definition of equivalence, $u \sim v$ if and only if for each morphism $h: A^* \rightarrow M$ into a finite monoid, there exists an $N \in \N$ such that for all $n,m \geq N$ we have $d(u_n,v_m) < 2^{- \left| M \right|}$. By a similar argument to the one above, this is implies $F_u = F_v$.

By these observation, the mapping
\begin{align*}
\phi: \Ahat &\rightarrow \mathcal{S}(\Reg)\\
[u] & \mapsto F_u
\end{align*}
is well-defined.

We claim that this mapping is a bijection. To show, that it is surjective, let $\gamma \in \mathcal{S}(\Reg)$. The set
$$\{\overline{L} \mid L \in \gamma\}$$
is a collection of closed subsets of $\Ahat$ and, since $\gamma$ is an ultrafilter, has the finite intersection property. Furthermore, $\Ahat$ is a compact space, and thus
$$F = \bigcap_{L \in \gamma} \overline{L} \neq \emptyset.$$
Hence there exists a sequence $u$ such that $[u] \in F$. Let $L \in \gamma$. By definition of $F$, we obtain that $[u] \in \overline{L}$ and thus there exists an $n_0 \in \N$ such that for all $n \geq n_0$ we have $u_n \in L$. But this is exactly saying that $L \in F_u$. Hence $\gamma \subseteq F_u$ and since ultrafilters are maximal $\gamma = F_u$. Thus the mapping is surjective.

Furthermore, it is injective. Suppose that there exist two Cauchy sequences $u$ and $v$ with $[u] \neq [v]$. Then there exists a $L \in \Reg$ such that $[u] \in \overline{L}$ but $[v] \notin \overline{L}$. Since $\gamma$ is an ultrafilter, either $L$ or $L^c$ is an element of $\gamma$. Hence, either $[u] \in F$ or $[v] \in F$ and $F$ is a singleton. Thus, $\phi$ is a bijective mapping and there is a well-defined map
\begin{align*}
\psi: \mathcal{S}(\Reg) &\rightarrow \Ahat\\
\gamma &\mapsto \bigcap_{L \in \gamma} \overline{L}
\end{align*}
By construction, these maps are mutually inverse and $\widehat{L}  = \{\gamma \mid L \in \gamma\}$ gets mapped to $\overline{L}$ and vice versa. Thus, clopen sets get mapped to clopens. Both $\Ahat$ and $\mathcal{S}(\Reg)$ are known to be compact, Hausdorff and totally disconnected. These conditions are sufficient for the topology being uniquely determined by the clopen sets. Thus, the two spaces are homeomorphic.
\end{proof}

Take as an example the profinite word $a^{\omega}$ when seen as an element of $\mathcal{S}(\Reg)$
$$a^{\omega} = \{L \in \Reg \mid \exists n_0 \in \N \quad \forall n \geq n_0 : a^{n!} \in L\}.$$

Another Stone space that is of special interest, is the dual of the powerset of a set $X$. It is called the \emph{Stone-\v{C}ech} compactification of $X$ and is denoted by $\beta X$. An important property of $\beta X$ is that any map $f: X \rightarrow K$ into a compact Hausdorff space $K$, can be extended uniquely to a continuous map $\beta f: \beta X \rightarrow K$. Furthermore any map $f: X \rightarrow Y$ has a unique continuous extension, also denoted by $\beta f: \beta X \rightarrow \beta Y$. It is defined by the equivalence
$$L \in \beta f (\gamma) \Leftrightarrow f^{-1}(L) \in \gamma$$
for all $L \in \mathcal{P}(X)$ and $\gamma \in \beta X$.

Another useful property that originates from Stone duality is that if $\mathcal{B}$ is a Boolean algebra and $\mathcal{C}$ is a subalgebra of $\mathcal{B}$, then the associated Stone space $\mathcal{S}(\mathcal{C})$ is a quotient of $\mathcal{S}(\mathcal{B})$. If there is an embedding of $\mathcal{C}$ in $\mathcal{B}$ then there is a surjective map between the Stone spaces.
\begin{align*}
\mathcal{C} & \rightarrowtail \mathcal{B}\\
\mathcal{S}(\mathcal{C}) & \hookleftarrow \mathcal{S}(\mathcal{B})
\end{align*}

Since any Boolean algebra $\mathcal{B}$ is a subalgebra of the powerset $\mathcal{P}(X)$, we know that its Stone space $\mathcal{S}(\mathcal{B})$ will be a quotient of $\beta(X)$. 
Note that an ultrafilter of $\mathcal{C}$ when seen as a collection of sets is not an ultrafilter of $\mathcal{B}$, since the fourth property is violated. However, it still satisfies conditions $1.-3.$. A collection of sets having these properties is called a \emph{filter} of $\mathcal{B}$. As a special case of this, any profinite word, may be seen as a filter of $\beta A^*$. Again, for example, $a^{\omega}$ may be seen as an ultrafilter of $\mathcal{S}(\Reg)$ and thus is a filter of $\beta A^*$.
Hence, any point of $\mathcal{S}(\mathcal{C})$ defines a filter on $\mathcal{B}$. Furthermore, since ultrafilters are maximal filters with respect to inclusion, by Zorn's Lemma, any filter can be extended to an (non-unique) ultrafilter.

An even weaker property holds. It will often be necessary to show that there exist two ultrafilters having certain properties. For these properties it will be sufficient to know that the filter contains a specific collection of subsets. A collection of subsets is said to have the finite intersection property, if any two sets have non-empty intersection. We call such a set a \emph{filterbase}. By the same argument, any filterbase can be extended to an ultrafilter. We will use this method to construct ultrafilters from filterbases to obtain ultrafilters with certain properties.

Furthermore, the set of all filters of $\mathcal{B}$, denoted by $\mathcal{F}(\mathcal{B})$ is isomorphic to the \emph{Vietoris} space of $\mathcal{S}(\mathcal{B})$, the space of all closed subsets of $\mathcal{S}(\mathcal{B})$, which is denoted by $\mathbb{V}(\mathcal{S}(\mathcal{B}))$.

\subsection{Equations}\label{apx:equations}

Let $\mathcal{B}$ be a Boolean algebra and $\mathcal{C}$ a subalgebra of $\mathcal{B}$. A $\mathcal{B}$-equation is a tuple $(\mu,\nu)$ of ultrafilters on $\mathcal{B}$. We say that $\mathcal{C}$ satisfies the $\mathcal{B}$-equation $(\mu,\nu)$, if the following equivalence holds.

$$L \in \mu \Leftrightarrow L \in \nu$$

for all $L \in \mathcal{C}$. For a more convenient notation, we write that $\mathcal{C}$ satisfies the equation $[\mu \leftrightarrow \nu]$. Since $\mathcal{S}(\mathcal{C})$ is a quotient of $\mathcal{S}(\mathcal{B})$, it is defined by a set of $\mathcal{B}$-equations. When specializing this to the Boolean algebra of regular language, one obtains the following result.

\begin{theorem}[\cite{GGP}]
Every Boolean algebra of regular languages is defined by a set of profinite equations of the form $[u \leftrightarrow v]$, where $u$ and $v$ are profinite words.
\end{theorem}

Gehrke, Grigorieff and Pin showed more generally.

\begin{theorem}[\cite{GGP}]
Every Boolean algebra of languages is defined by a set of ultrafilter equations of the form $[\mu \leftrightarrow \nu]$, where $\mu$ and $\nu$ are ultrafilters on the set of words.
\end{theorem}

Note that the first result is similar to the famous theorem of Reiterman.

\begin{theorem}[\cite{reitermanTheorem}]
Every variety of languages is defined by a set of profinite equations.
\end{theorem}

This theorem is stronger in a sense, that the equations are not dependant on the underlying alphabet. Since varieties are closed under taking inverse morphisms, there is one set of equations, that suffices to describe the whole variety and not just the Boolean algebra assigned to a fixed alphabet.

In the profinite case, equations expose a strong connection to morphisms. Let $\VMon$ be a variety of monoids and $\LangVar_A$ the corresponding Boolean algebra. Then $\LangVar_A$ satisfies the profinite equation $u = v$ if and only if for each morphism $\varphi: A^* \rightarrow M$ into a monoid $M$ of $\VMon$, the equality $\hat{\varphi}(u) = \hat{\varphi}(v)$ holds.

\newpage

\section{Constant Size Circuits Families}\label{sec:circuits}

In this paper we consider circuits over arbitrary alphabets. In contrast to the usual notion of size we do not count input gates, and hence only call them inputs. Thus, the size of a circuit is the number of inner nodes.
This is necessary to allow circuit families with gates of unbounded fan-in to access inputs of arbitrary length and still have constant size.

\begin{definition}[Family of Boolean functions defined by a language]
A language $L\subseteq\{0,1\}^*$ in a natural way defines a family of Boolean functions, denoted by $f^L = (f^L_i)_{i\in\N}$ where $f^L_i(x_0,\dots,x_{i-1})=1$ iff $x_0\dots x_{i-1}\in L$.
\end{definition}

\begin{definition}
A \emph{variety of languages} is a class of languages $\mathcal{V}$ such that
\begin{enumerate}
\item for each alphabet $A$, $\mathcal{V}_{A}$ is a Boolean algebra over $A^*$
\item for each morphism $\varphi : A^* \rightarrow B^*$, the condition $L \in \mathcal{V}_{B}$ implies $\varphi^{-1}(L) \in \mathcal{V}_{A}$
\item for each  $L \in \mathcal{V}_{A}$ and $a \in A$, we have $a^{-1}L := \{w \in A^* \mid aw \in L\} \in \mathcal{V}_{A}$ and\\ $La^{-1} := \{w \in A^* \mid wa \in L\} \in \mathcal{V}_{A}$
\end{enumerate}
\end{definition}

Recall that a base is a set containing Boolean functions and families of Boolean functions.
To treat circuits in a more general way we will define a base defined by a variety of languages. This allows us to describe different constant size circuit families over arbitrary alphabets simply by considering different varieties.
The definition will allow a base to consist of an infinite number of Boolean families, but a circuit family over an (infinite) base is only allowed to use a finite subset of the elements in the base. 

\begin{definition}[Bases defined by a variety of languages]
Given a variety $\mathcal V$ of regular languages, $\mathcal V_{\{0,1\}}$ is a collection of languages in $\{0,1\}^*$ and each of these languages defines a family of Boolean functions. We call the set $\{f^L \mid L \in \mathcal V_{\{0,1\}}\}$ the base defined by $\mathcal V$.
\end{definition}

For our purpose, it suffices to consider bases generated by varieties of languages, where the languages are regular and commutative. This is not much of a limitation as many gate types correspond to commutative regular languages (see table on the left of Figure \ref{fig:gatetype}).

\begin{figure}[h!]
\begin{tabular}{ll}
\toprule
Gate Type & Language \\ 
\midrule
$\land$ & $1^*$ \\
$\lor$  & $\{0,1\}^*1\{0,1\}^*$ \\
$\modp$  & $\{0,1\}^*((10^*)^p)^*$ \\
\bottomrule
\end{tabular}
\hfill
\begin{tabular}{ll}
\toprule
Circuit class & Base \\ 
\midrule
$\ACzero$ & $\{\land,\lor\}$ \\
$\CCzero$ & $\{\modp \mid p\in\N\}$\\
$\ACCzero$ & $\{\land,\lor,\modp \mid p\in\N\}$ \\
\bottomrule
\end{tabular}\\[3pt]
\caption{On the left: Typical gate types and the languages they are defined by. On the right: Typical circuit classes}\label{fig:gatetype}
\end{figure}

The following lemma implies that for each circuit over a base generated by a variety of regular commutative languages $\LangVar$, there exists a circuit that accepts exactly the same languages and can be written as one layer of gates from $\LangVar$ accessing the inputs and below an $\NCzero$ circuit.

\begin{lemma}\label{lem:circuitlemma}
Let $\CircBase$ be a base generated by a variety $\LangVar$ of regular commutative languages. A language $L$ is recognized by a constant size circuit family over $\CircBase$ if and only if it is recognized by a constant size circuit family where the gates only accessing inputs are from $\CircBase$, and all other gates are labeled by $\land_2$ or $\lor_2$.
\end{lemma}

\begin{proof}\label{proof:circuitlemma}
Assume $(C_i)_{i\in\N}$ has a size bound of $k$.
For each $i$ we generate a new circuit $C'_i$. In a topological ordered way we proceed for each gate $g\in C_i$ and construct a corresponding gate in $C'_i$.

If the gate $g$ is labeled by a $l$-ary Boolean function, then it is clear that the value can be determined by the value of a bounded number of gates or inputs that by the topological order are already in $C'_i$.

Otherwise it takes more effort to show this. Assume $g$ is labeled with a Boolean family $f$. As the base is generated by a variety $\LangVar$, there is a language $L\in\LangVar_{\{0,1\}}$ that corresponds to $f$. Let $L_1,\dots,L_l$ be all two sided quotients of $L$ which implies that $L_i\in\LangVar_{\{0,1\}}$ and hence the corresponding Boolean family $f^{L_i}$ is also in the base. For each $f^{L_i}$ we have a gate with this label in $C'_i$ that is connected in the same way to the inputs as $g$ to the corresponding inputs in $C_i$. By the topological order all preceding gates of $g$ have already corresponding gates in $C'_i$, and since $C_i$ is bounded by size $k$ there are at most $k-1$ such gates. As $L$ is a commutative regular language the value of $g$ can be determined by the value of the newly generated $l$ gates and the preceding at most $k-1$ gates (This is because the two sided quotients actually allow us to determine the element of the syntactic monoid that this part of the input computes to. For the other at most $k-1$ this is also possible as they contribute exactly one letter. Hence containment in $L$ or the question if $f$ evaluates to $1$ depends only on the product of these monoid elements). Hence it can be determined by the value of a bounded number of gates in $C'_i$. (Note that $l$ is bounded since the base is finite and hence we only consider a finite number of languages $L$).

As in both cases the value only depends on a bounded number of gates we can add a Boolean circuit with $\land_2$ and $\lor_2$ gates to $C'_i$ that evaluates this number. The gate that computes this number is the corresponding gate to $g$.

\end{proof}

\section{The Block Product for Varieties of Languages}\label{sec:blockproduct}

In the last section, we introduced bases for circuits that were defined by a variety of languages. 
Here we will define an unary operation $\cdot\block\Parb$ on varieties mapping a variety of commutative regular languages $\mathcal V$ to the variety of languages $\mathcal V\block\Parb$ recognized by constant size circuits over the base $\mathcal V$. 

This rather strange looking notation comes from the algebraic background where similar ideas have been used on the algebraic side in \cite{WBP1}. Using the algebraic tools from that paper one could show that constant size circuit families recognize the same languages as the finitely typed groups in the block product $\mathbf V\block\Parbbf$, where $\mathbf V$ are the (typed) monoids corresponding to the gate types and $\Parbbf$ are the typed monoids corresponding arbitrary predicates and hence to the non-uniform wiring of the circuit family.
In this paper however, we omit the algebraic definition of the block product of (typed) monoids but rather define a mechanism, that provides us with the same languages as recognized by the block product, and that is purely defined on the language side. For more details on the algebraic and logic side for varieties of regular languages we refer to a survey about the block product principle \cite{TeTh07} or for the non-regular case to \cite{KLR}.

Here we will restrict to the unary operation $\cdot\block\Parb$ suitable for our constant size circuit classes.

There is a natural morphism $|\cdot|:A^*\rightarrow\N$ that maps each word to its length. We say a mapping $f \colon A^* \rightarrow B^*$ is length preserving, if $\left| f(u) \right| = \left| u \right|$ for all $u \in A^*$.

\begin{definition}[$\N$-transduction]
Let $\mathcal D$ be a finite partition of $\N^2$. By $[(i,j)]_{\mathcal D}$ denote the equivalence class that $(i,j)$ belongs to.
Then a $\N$-transduction is a length preserving map $\tau_{\mathcal D}:A^*\rightarrow (A\times\mathcal D)^*$, where $(\tau_{\mathcal D}(w))_i=(w_i,[(|w_{<i}|,|w_{>i}|)]_\mathcal D) = (w_i,[(i,\left| w \right|-i-1)]_{\mathcal D})$.
\end{definition}

Finally we can use these transductions to define the block product. We only define a unary operation that maps a variety $\mathcal V$ to a variety $\mathcal V\block\Parb$. This rather strange notation stems from the strong connection of $\Parb$ with $\N$-transducers and should be just considered notation here.

\begin{definition}[$\mathcal V\block \Parb$]
Let $\mathcal V$ be a variety of languages, then we define $\mathcal V\block \Parb$ as the variety of languages, where $(\mathcal V\block \Parb)_{A}$ is generated by the languages $\tau_{\mathcal D}^{-1}(L)$ for all partitions $\mathcal D$ of $\N^2$, and $L\in V_{A\times\mathcal D}$.
\end{definition}

Because of the connection to the block product we call the languages $\mathcal V\block \Parb$.

\begin{lemma}\label{lem:taucircuit}
The languages in constant size circuits over a base defined by the variety $\mathcal V$, are exactly the languages in $\mathcal V\block\Parb$.
\end{lemma}

\begin{proof}\label{proof:taucircuit}
Let $L \subseteq A^*$ be a language recognized by a constant size circuit. Then $L$ is recognized by Lemma \ref{lem:circuitlemma} by a Boolean combination of depth 1 circuits with gates from $\mathcal V_{\{0,1\}}$. As $(\mathcal V\block\Parb)_A$ is a Boolean algebra it suffices to show that every language recognized by a depth 1 circuit with gates from $\mathcal V$ is in $(\mathcal V\block\Parb)_A$. So we assume that $L$ is recognized by a circuit of depth 1 which is just a single gate. Let $L'$ be the language corresponding to the function computed by this single gate. As the functions of the gates are symmetric we can assume that the gate queries the inputs in order. Also as $L'$ is regular the multiplicity of edges querying the same input position to the gate is limited by some constant. Let $c$ be this constant. Hence we can upper bound the different ways an input position is wired to this gate by $(2^{|A|})^c$. Let $\mathcal D$ be a partition where the equivalence classes correspond to the different ways the input can be wired. We define a morphism $h \colon (A\times\mathcal D)\rightarrow \{0,1\}^*$ where each $(a,P)$ is mapped to the way an input in the equivalence class $P$ reading the letter $a$ as input would influence the gate. As $\mathcal V$ is a variety $L''=h^{-1}(L')$ is in $\mathcal V_{(A \times \mathcal D)}$. But then $\tau_\mathcal D^{-1}(L'')=L$.

For the other direction as constant size circuits over the base generated by $\mathcal V$ are closed under Boolean combinations it suffices to show that any language $L \subseteq A^*$ with $L=\tau_\mathcal D^{-1}(L')$ and $L'\in\mathcal V_{(A \times \mathcal D)}$ is recognized by a constant size circuit. As $L'\in\mathcal V_{(A \times \mathcal D)}$ is a symmetric regular language it is a Boolean combination of languages $L'_1,\dots,L'_k\in\mathcal V_{(A \times \mathcal D)}$, such that there exists a morphism $h_i \colon (A \times \mathcal D)^* \rightarrow \{0,1\}^*$ and $L'_i=h_i^{-1}(L''_i)$, where $L''_i \in \mathcal V_{\{0,1\}}$. Fix an input length $n$. We construct a circuit that consists of exactly this Boolean combination of gates $g_1,\dots,g_k$ computing the functions corresponding to $L''_1,\dots,L''_k$. We wire each input to the gates such that for a word $w \in A^n$ each input position $j$ with $w_j=a$ contributes to the gate $g_i$ the value $h_i((\tau_\mathcal D(w))_j)$ and some neutral string otherwise. Please note that this definition of the wires is only well defined as $(\tau_\mathcal D(w))_j$ does only depend on the $w_j$, the length of $w$ and the position $j$, but not on the other letters. 
This completes our construction of the circuit.
\end{proof}

\section{Duality and the block product}\label{sec:duality}

We briefly introduce some theory from Stone Duality, which will be used to characterize the classes of languages we are interested in and to obtain separation results for them. A short introduction on the idea behind the theory is stated here.

By duality, each Boolean algebra $\B$ has an associated compact space, called its Stone Space $\mathcal S ( \B )$. For any two Boolean algebras $\B$ and $\mathcal C$, if $\mathcal C$ is a subalgebra of $\B$, a relation between the Stone spaces exists, namely $\mathcal S (\mathcal C)$ is a quotient of $\mathcal S (\B)$. Since every Boolean algebra of languages is a subalgebra of the powerset of $A^*$, there always exists a canonical projection from the Stone space of $P(A^*)$ to the Stone space of any Boolean algebra of languages over $A^*$. The idea is to characterize the Boolean algebra of languages $\B$ by the kernel of said projection, that is finding all pairs of elements in the Stone space of $P(A^*)$ that get identified in $\mathcal S(\B)$.

In order to define the points of the Stone space we need to define the notion of an ultrafilter. 

\begin{definition}[(Ultra)Filter]
Let $\B$ be a Boolean algebra. An \textit{filter} of $\B$ is a non-empty subset $\gamma$ of $\B$ that satisfies
\begin{enumerate}
\item $\emptyset \notin \gamma$,
\item if $L \in \gamma$ and $K \supseteq L$, then $K \in \gamma$, \hfill ($\gamma$ is closed under extension)
\item if $L,K \in \gamma$ then $K \cap L \in \gamma$, \hfill ($\gamma$ is closed under finite intersections)
\end{enumerate}
The filter is called {\em ultrafilter} if it additionally satisfies 
\begin{enumerate}
\setcounter{enumi}{3}
\item for each $L \in \B$, either $L \in \gamma$ or $L^c \in \gamma$. \hfill (ultrafilter condition)
\end{enumerate}
\end{definition}

\begin{definition}[Stone Space]
Let $\mathcal B$ be a Boolean algebra. The Stone space $\mathcal S( \mathcal B)$ of $\B$ is the space of all ultrafilters of $\mathcal B$ equipped with the topology generated by the sets $\widehat{L} = \{\gamma \in \mathcal{S}(\B) \mid L \in \gamma \}$ for $L \in \B$.
\end{definition}

The topology that the Stone spaces are equipped with is of importance, since it holds informations about the languages in the underlying Boolean algebra. For those familiar with topology: The clopen sets of $\mathcal S( \B)$ are exactly the topological closures of the sets $L\in\mathcal B$.

The Stone Space of the full Boolean algebra $P(A^*)$ is a special case, also known as the Stone-\v{C}ech compactification of $A^*$, which is denoted by  $\beta (A^*)=\mathcal S ( \mathcal P (A^*))$. For an algebra $\mathcal B\subseteq P(A^*)$, we denote the canonical projection from $\beta(A^*)$ onto $\mathcal{S}(\B)$ by $\pi_\B$. 
Let $A^*$ and $B^*$ be two free monoids and $f: A^* \rightarrow B^*$ be a function. Then there exists a unique continuous extension $\beta f: \beta A^* \rightarrow \beta B^*$, which is defined by $$L \in \beta f(\gamma) \Leftrightarrow f^{-1}(L) \in \gamma.$$ See \cite{something}.

\begin{definition}[Equation]
An ultrafilter equation is a tuple $(\mu,\nu) \in \beta A^* \times \beta A^*$. Let $\B$ be a boolean algebra. We say that $\B$ satisfies the equation $(\mu,\nu)$ if $\pi_\B(\mu)=\pi_\B(\nu)$. With respect to some Boolean algebra $\B$ we say that $[\mu \leftrightarrow \nu]$ holds.
\end{definition}

\begin{lemma}\label{lem:projection}
Let $\B$ be a subalgebra of $P(A^*)$.
For $\mu,\nu\in\beta A^*$ we have $\pi_\B(\mu)=\pi_\B(\nu)$ iff for all $L \in \B$ the equivalence
$L \in \mu \Leftrightarrow L \in \nu$
holds. 
\end{lemma}

\begin{proof}
The projection $\pi_\B$ is given by $\pi_\B(\mu) = \{L \in \mu \mid L \in \B\}$ and thus the equivalence holds.
\end{proof}

Recently, Gehrke, Grigorieff and Pin \cite{GGP} were able to show that any Boolean algebra of languages can be defined by a set of equations of the form $[\mu \leftrightarrow \nu]$, where $\mu$ and $\nu$ are ultrafilters on the set of words. That is $L \in \B$ if and only if for all equations $[\mu \leftrightarrow \nu]$ of $\B$ the equivalence
$L \in \mu \Leftrightarrow L \in \nu$
holds.
We say a set of equations is sound, if all $L$ in $\B$ satisfy the equivalence above and complete, if a language in $A^*$ satisfying all equations is in $\B$. 

This theorem provides us with the existence of ultrafilter equations for $(\BlockP)_{A}$. However, it does not answer the question on how to obtain them. The following lemma provides us with a set of equations that define precisely the kernel of the projection  $\pi_{(\mathcal{V} \block \Parb)_{A}}$. It builds on the knowledge, that $(\BlockP)_{A}$ was defined by the functions $\tau_{\mathcal D}$.

\begin{lemma}\label{lem:equations}
Let $\mu,\nu \in \beta A^*$. Then for each partition $\mathcal D$ of $\N^2$ the Boolean algebra $\mathcal V_{(A \times \mathcal D)}$ satisfies the equation
$[\beta \tau_{\mathcal D}(\mu) \leftrightarrow \beta \tau_{\mathcal D}(\nu)]$
if and only if $[\mu \leftrightarrow \nu]$ is an equation of $(\BlockP)_{A}$.
\end{lemma}
\begin{proof}
Let $\mu,\nu \in \beta A^*$ such that $[\beta \tau_{\mathcal D}(\mu) \leftrightarrow \beta \tau_{\mathcal D}(\nu)]$ holds for all partitions $\mathcal D$ of $\N^2$ and let $L \in \BlockPA$ be a generator of the Boolean algebra. Recall that by definition there exists a partition $\mathcal D$ of $\N^2$ and a language $S \in \mathcal V_{(A \times \mathcal D)}$ such that $L = \tau_{\mathcal D}^{-1}(S)$. Then

$$ L \in \mu \Leftrightarrow  \tau_{\mathcal D}^{-1}(S) \in \mu \Leftrightarrow  S \in \beta \tau_{\mathcal D}(\mu) \Leftrightarrow  S \in \beta \tau_{\mathcal D}(\nu) \Leftrightarrow  \tau_{\mathcal D}^{-1}(S) \in \mu \Leftrightarrow  L \in \nu.$$

\noindent This proves both directions of the claim.

\end{proof}

This set of equations already provides us with a full characterization of $(\BlockP)_{A}$, but we are interested in a set that satisfies conditions that are easier to check and still is sound and complete.

\section{Equations for the block product}\label{sec:equations}
In this chapter we find a set of equations that holds for $(\BlockP)_{A}$, depending on the equations that define the variety $\mathcal V$ of the gate types, which in our case is regular and commutative. As a corollary, we expose separation results for a selection of classes, to demonstrate the applicability of the equations.

As we describe the base of the circuits by a variety $\mathcal V$ of regular languages, we use a description that has already been applied in the regular case. Thereto, we introduce the notion of identities of profinite words. Here, we define a profinite word as an ultrafilter on the regular languages.
The combined results of Reiterman \cite{reitermanTheorem} and Eilenberg \cite{Eilenberg:1976:ALM:540244} state that for each variety of regular languages $\mathcal V$, there is a set of profinite identities, defining the variety. Informally speaking, an identity is an equation that holds not only for a language, but also for all quotients of a language.

As such, we can define the notion of profinite identities in the following way: A Boolean algebra of regular languages $\B$ satisfies the profinite identity $[u = v]$, where $u,v \in \mathcal S(\Reg)$ instead of $u,v \in \beta(A^*)$, if for all $L \in \B$ the equivalence
$x^{-1}Ly^{-1} \in u \Leftrightarrow x^{-1}Ly^{-1} \in v$
for all $xy \in A^*$ holds.
As varieties are closed under quotients, it suffices to consider $L \in u \Leftrightarrow L \in v$ for all $L \in \B$.

To define the equations that hold for $(\BlockP)_A$, we define a function, that gets as arguments a word $w$, another word $s$ and a vector of positions $p$, such that the positions of $p$ in $w$ are substituted by the letters of $s$. Naturally, such a substitution only makes sense, if the input is restricted to be reasonable. For instance the positions in $p$ should not exceed the length of the word $w$. For technical reasons, each element of the vector of positions will be a tuple containing the distance of the position from the beginning and the end of the word. 

For an element $p \in \N^2$ denote by $p^1$ the first and by $p^2$ the second component of $p$, i.e. $p = (p^1,p^2)$.
We define the set of correct substitutions
$$
\mathfrak{D} = \left\{ (\w,\s,\p) \in A^* \times A^* \times (\N^2)^* \left\vert \parbox[c]{3.5cm}{$\left|\s\right|=\left|\p\right| \\ \forall i \colon \left| w \right| -1 = p_i^1 + p_i^2 \\ \p^1_0<\ldots<\p^1_{\left|\p\right|-1}<\left|\w\right|$}\right. \right\}.
$$
Given a word $w = w_0 \ldots w_{m-1} \in A^*$ of length $m$ and $k,l$ with $0 \leq k \leq l < m$, we define $w[k,l] = w_k \ldots w_l$.
Let $n$ be the length of $s$, then the function $f: A^* \times A^* \times (\N^2)^* \rightarrow A^*$ is defined as
\[
f(\w,\s,\p) = 
\begin{cases}
\w[0,\p^1_0-1] \s_0 \w[\p^1_0+1,\p^1_{1}-1]\s_1 \ldots \s_{n-1} \w[\p^1_{n-1}+1,m-1] & \text{if $(\w,\s,\p) \in \mathfrak{D}$,}\\
w & \text{otherwise.} 
\end{cases}
\]

\noindent Furthermore, define the function that maps the second component to its length as
\begin{align*}
\lambda \colon A^* \times A^* \times (\N^2)^* & \rightarrow A^* \times \N \times \N^*\\
(w,s,p) & \mapsto (w, \left|s \right|, p)
\end{align*}
and let $\pi_2 \colon A^* \times A^* \times (\N^2)^* \rightarrow A^*$ with $\pi_2(w,s,p) = s$ be the projection on the second component.

As the function $f$ substitutes letters in certain positions, we need the following definition in order to define which positions of $p$ are ``indistinguishable'' by a language in $(\BlockP)_A$.

Define the mapping $\contPos \colon A^* \times A^* \times (\N^2)^* \rightarrow \mathcal{P}(\N^2)$ that maps the third component onto its content, given by $\contPos(w,s,p) = \{p_0,\ldots,p_{\left|p\right|}\}$.
Note that any finite subset of $\N^2$ is a finite subset of $\beta(\N^2)$ and thus $\pi_c$ can be interpreted as a mapping into the space $\mathcal F(\N^2)$ of all filters of $\mathcal{P}(\N^2)$, by sending it to the intersection of all ultrafilters containing the set.
Furthermore, $\beta (\N^2)$, which contains all ultrafilters of $\mathcal{P}(\N^2)$ can be seen as a subspace of $\mathcal{F}(\N^2)$, which is homeomorphic to Vietoris of $\beta(\N^2)$.
Then there exists and extension of $\pi_c$ denoted by $\beta \pi_c$, known as the Stone-\v{C}ech extension $\beta \pi_c \colon \beta(A^* \times A^* \times (\N^2)^*) \rightarrow \mathcal{F}(\N^2)$.

Together with these definitions we can formulate the theorem that provides us with a set of equations for $(\BlockP)_A$.

\begin{theorem}\label{thm:main}
The variety $\BlockPA$ is defined by the equations
\[[\beta f(\gammaU) \leftrightarrow \beta f(\gammaV)].\]
where $[u = v]$ is a profinite equation that holds on $\mathcal V$ and $\gammaU,\gammaV \in \beta(A^* \times A^* \times (\N^2)^*)$ satisfying
\begin{enumerate}
\renewcommand{\labelenumi}{\em(\arabic{enumi})}
\setlength{\itemindent}{8pt}
\item \label{cond1} $\beta \lengthPres(\gammaU) = \beta \lengthPres (\gammaV)$
\item \label{cond2} $u \subseteq \beta \pSub(\gammaU) \text{ and } v \subseteq \beta \pSub(\gammaV)$
\item \label{cond3} $\beta \contPos (\gammaU) = \beta \contPos (\gammaV) \in \beta(\N^2)$
\end{enumerate}
\end{theorem}
\begin{proof}
We prove soundness of these equations in Section \ref{sec:soundness} and completeness in Section \ref{sec:completeness}
\end{proof}

While the following separation result itself is not surprising, and strong separations are known, the proof method has the advantage that there is no need for probabilistic methods to find specific inputs for the circuits to be fixed or swapped.

For a fixed $x \in A^*$ define the profinite word, also denoted by $x$ as
$$x = \{L \in \Reg \mid x \in L\}$$
and $x^{\omega}$ as
$$x^{\omega} = \{L \in \Reg \mid \exists n_0 \in \N \quad \forall n \geq n_0 \quad \colon \quad x^{n!} \in L\}$$

The following varieties are used for characterization of $\BlockP$ by equations.
\begin{figure}[h!]
  \begin{tabular}{lll}
  \toprule
   Gates & Profinite Identities & Circuit Class (constant size) \\ 
\midrule
   $\{\land,\lor\}$ & $xy = yx \quad x^2y = xy^2$ & $\ACzero$ \\
   $\{\land,\lor,\modp \mid p\in\N\}$ & $xy = yx$  & $\ACCzero$ \\ 
    $\{\modp \mid p\in\N\}$ & $xy = yx \quad x^{\omega} = y^\omega$ & $\CCzero$ \\
\bottomrule
  \end{tabular}
\caption{Varieties defining the gate types and the defining profinite identities.}\label{fig:varieties}
\end{figure}

\begin{corollary}
Constant size $\mathbf{CC}^0$ is strictly contained in constant size $\mathbf{ACC}^0$ and constant size $\mathbf{AC}^0$ is strictly contained in constant size $\mathbf{ACC}^0$.
Also constant size $\mathbf{CC}^0$ and constant size $\mathbf{AC}^0$ are not comparable.
\end{corollary}
\begin{proof}
We show this by proving $L_{\text{AND}}=1^*$ is not contained in constant size $\CCzero$ and parity $L_{\text{PARITY}}=(0^*10^*1)^*0^*$ is not contained in constant size $\ACzero$.

For that we construct ultrafilters from filterbases, using Lemmata \ref{lemma:pullback}, \ref{lemma:addPullback} and \ref{theorem:contPos} from Section \ref{sec:completeness}, such that they satisfy conditions 1.-3. from Theorem \ref{thm:main}.

Take the identity $[0^\omega=1^\omega]$, that holds for the variety providing us with the gate types of $\CCzero$. By Theorem \ref{thm:main}, we know that for any two ultrafilters $\gamma_{0^{\omega}}$ and $\gamma_{1^{\omega}}$ satisfying the conditions of the Theorem, provide us with an equation $[\beta f(\gamma_{0^\omega}) \leftrightarrow \beta f(\gamma_{1^\omega})]$. In constructing two such filters, such that $L_{\text{AND}}$ is contained in $\beta f(\gamma_{0^\omega})$, but not in $\beta f(\gamma_{0^\omega})$, we prove that it is not an element of constant size $\CCzero$.
Consider the filter base
$$\mathcal F_1 = \{A^* \times \{n! \mid n \geq N\} \times (\N^2)^* \mid N \in \N\}$$
and the second base
$$\mathcal F_2 = \{\bigcup_{i=0}^n A^* \times \N \times P_i^* \mid \{P_0,\ldots,P_n\} \text{ is a partition of } \N^2\}.$$
Adding the two together yields another filterbase, denoted by $\mathcal F$, as none of the elements have empty intersection. Let $\mu \in \beta(A^* \times \N \times \N^*)$ be an ultrafilter containing the filter base $\mathcal F$.
Next, consider the set
$$1^* \times \{1^{n!} \mid n \in \N\} \times (\N^2)^*.$$
By definition of $f$, we obtain $f(1^* \times \{1^{n!} \mid n \in \N\} \times (\N^2)^*) \subseteq L_{\text{AND}}$ and thus $1^* \times \{1^{n!} \mid n \in \N\} \times (\N^2)^* \subseteq f^{-1}(L_{\text{AND}})$.
Adding this to the pullback by $\lambda^{-1}(\mu)$ yields another filter base, denoted by $\mathcal F_{1^\omega}$. By $\mathcal F_{0^\omega}$ we denote the base $\lambda^{-1}(\mu)$ when adding the set
$$1^* \times \{0^{n!} \mid n \in \N\} \times (\N^2)^*.$$
Let $\gamma_{0^{\omega}}$ be an ultrafilter containing $\mathcal F_{0^\omega}$ and $\gamma_{1^{\omega}}$ be an ultrafilter containing $\mathcal F_{1^\omega}$. Then both ultrafilters satisfy conditions 1.-3. of the Theorem, such that $[\beta f(\gamma_{0^\omega}) \leftrightarrow \beta f(\gamma_{1^\omega})]$ holds for $\CCzero$. But $L_{\text{AND}} \in \beta f(\gamma_{1^\omega})$ and $L_{\text{AND}} \notin \beta f(\gamma_{0^\omega})$. Hence $L_{\text{AND}}$ is not in constant size $\CCzero$.

Equivalently, we use the identity $[110=100]$ satisfied by the variety corresponding to the gate types of $\ACzero$.
Again let
$$\mathcal F_2 = \{\bigcup_{i=0}^n A^* \times \N \times P_i^* \mid \{P_0,\ldots,P_n\} \text{ is a partition of } \N^2\}.$$
Adding the set $A^* \times \{3\} \times \N^*$ yields a filter base $\mathcal F'$. Let $\nu \in \beta(A^* \times \N \times \N^*)$ be an ultrafilter containing $\mathcal F'$. Consider the sets
$$S_{110} = 0^* \times \{110\} \times (\N^2)^* \text{ and } S_{100} = 0^* \times \{100\} \times (\N^2)^*.$$
Adding the set $S_{110}$ to the pullback $\lambda^{-1}(\nu)$ provides us with a new filter base $\mathcal F_{110}$ and respectively adding $S_{100}$ to $\lambda^{-1}(\nu)$ with a filter base $\mathcal F_{100}$.
Let $\gamma_{110}$ be an ultrafilter containing $\mathcal F_{110}$ and $\gamma_{100}$ be an ultrafilter containing $\mathcal F_{100}$. Then both ultrafilters satisfy conditions 1.-3. of the main theorem and thus $[\beta f(\gamma_{110}) \leftrightarrow \beta f(\gamma_{100})]$ is an equation satisfied by constant size $\ACzero$.
Since $f(S_{110}) \subseteq L_{\text{PARITY}}$ and $f(S_{100}) \subseteq L_{\text{PARITY}}^c$, we obtain $L_{\text{PARITY}} \in \beta f(\gamma_{110})$ but $L_{\text{PARITY}} \notin \beta f(\gamma_{100})$ and thus it is not in constant size $\ACzero$.
\end{proof}


\section{Proof of Soundness}\label{sec:soundness}
Let $u,v$ be two profinite words such that $[u = v]$ is an identity for $\mathcal V$ and let $\gamma_u, \gamma_v \in \beta(A^* \times A^* \times (\N^2)^*)$ be the associated ultrafilters satisfying the conditions of Theorem \ref{thm:main}. We have to show that
$[\beta f(\gamma_u) \leftrightarrow \beta f(\gamma_v)]$
is an equation of $(\BlockP)_{A}$. By Lemma \ref{lem:equations} this is equivalent to
$$[\beta \tau_{\mathcal D}(\beta f (\gamma_u)) \leftrightarrow \beta \tau_{\mathcal D}(\beta f (\gamma_v))]$$
being an equation of $\mathcal V$ for each partition $\mathcal D$ of $\N^2$. Thus, by Lemma \ref{lem:projection} it suffices to show that for each partition $\mathcal D$ and each language $L \in \mathcal V_{(A \times \mathcal D)}$, the equivalence
$$f^{-1}(\tau_{\mathcal D}^{-1}(L)) \in \gamma_u \Leftrightarrow f^{-1}(\tau_{\mathcal D}^{-1}(L)) \in \gamma_v$$
holds.

For every $P\in\mathcal D$ we have a morphism $h_P:A^*\rightarrow (A\times\mathcal D)^*$ that maps $a\mapsto (a,P)$.
Let $\mathcal E$ be the smallest boolean algebra containing all languages $h_P^{-1}(x^{-1}L)$ where $x \in (A\times\mathcal D)^*$.
As $\mathcal V$ is a variety and thus closed under inverse morphisms and quotients, $\mathcal V_A$ contains all languages in $\mathcal E$.
The minimal elements, with respect to inclusion, of $\mathcal E$ form a partition of $A^*$. Hence $u$ and $v$ each contain at one of these partition elements (see Lemma \ref{lem:partitionimappendix}. Since the minimal elements are languages of $\mathcal V$ and $[u=v]$ is an equation of $\mathcal V$, $u$ and $v$ actually contain the same partition element. Let $E_{u,v}\in\mathcal V$ denote this partition element by $E_{u,v}\in u,v$.  As  $\pi_2(\gamma_u)$ and $\pi_2(\gamma_v)$ extend $u$ and $v$ both of them contain $E_{u,v}$.

To show $f^{-1}(\tau_{\mathcal D}^{-1}(L))$ is contained in $\gamma_u$ if and only if it is contained in $\gamma_v$, by upward closure of filters it is enough to consider the subset $f^{-1}(\tau_{\mathcal D}^{-1}(L)) \cap \pi_2^{-1}(E_{u,v})$. 

This subset can be rewritten as the following join decomposition:
$$ (f^{-1}(\tau_{\mathcal D}^{-1}(L)) \cap \pi_2^{-1}(E_{u,v}) \cap \mathfrak{D}^c) \cup \left( \bigcup_{P \in \mathcal D} (\pi_c^{-1}(P) \cap f^{-1}(\tau_{\mathcal D}^{-1}(L)) \cap \pi_2^{-1}(E_{u,v}) \cap \mathfrak{D})\right),$$
where $\mathfrak{D}$, as defined in Section \ref{sec:equations}, denotes the the points where the function $f$ is not just a projection on the first component.
Since $\lambda^{-1}(\lambda(\mathfrak{D})) = \mathfrak{D}$ by condition \cond1 we obtain $\mathfrak{D} \in \gamma_u \Leftrightarrow \mathfrak{D} \in \gamma_v.$

We show for each set of the join separately that it is contained in $\gamma_u$ if and only if it is contained in $\gamma_v$. This is sufficient, as both $\gamma_u$ and $\gamma_v$ are ultrafilters and the sets of the decomposition are disjoint.

\begin{itemize}
\item The set $f^{-1}(\tau_{\mathcal D}^{-1}(L)) \cap \pi_2^{-1}(E_{u,v}) \cap \mathfrak{D}^c$:
\end{itemize}

\noindent By definition of $f$, when restricted to $\mathfrak{D}^c$, it is the identity of on the first component. We obtain
$$f^{-1}(\tau_{\mathcal D}^{-1}(L)) \cap \pi_2^{-1}(E_{u,v}) \cap \mathfrak{D}^c = \pi_1^{-1}(L) \cap \pi_2^{-1}(E_{u,v}) \cap \mathfrak{D}^c$$ and since $\pi_1$ factors through $\lambda$, by condition \cond1 that this term belongs to $\gamma_u$ if and only if it belongs to $\gamma_v$.

\begin{itemize}
\item The set $f^{-1}(\tau_{\mathcal D}^{-1}(L)) \cap \pi_2^{-1}(E_{u,v}) \cap \pi_c^{-1}(P) \cap \mathfrak{D}$:
\end{itemize}

\noindent Let $P \in \mathcal D$ be a partition element. The above set then, by definition, is equal to
$$\{(w,s,p) \in \mathfrak{D} \mid \tau_{\mathcal D}(f(w,s,p)) \in L, \pi_c(w,s,p) \in P \text{ and } s\in E_{u,v}\}.$$

\noindent For a word $w \in A^*$ and $Q \subseteq (\N^2)$ with $Q = \{p_0,\ldots,p_k\}$ where $p^1_1 < p^1_2 < \ldots < p^1_k$ define
$$w[Q] = w_{p_1^1}\ldots w_{p_k^1}$$

\noindent The condition we need to examine more closely is $\tau_{\mathcal D}(f(w,s,p)) \in L$. Therefore, we need to have a closer look at $\tau_{\mathcal D} \circ f$.
The requirement $\pi_c(w,s,p) \in P$ provides us with the information of all letters of $p$ being in the same equivalence class $P$ of $\mathcal D$ and by $(w,s,p)\in\mathfrak{D}$, we obtain $|p|=|s|$.
We split the image of  $\tau_{\mathcal D} \circ f$ into the part of the substitution and the remaining part. Fix a tuple $(w,s,p)$, and let $S$ be the set of substituted positions, that is, $S = \{p_0,\ldots,p_{\left|p\right|-1}\}$. Let $r$ be the remaining part in $(A \times \mathcal D)^*$ defined by 
$$r = (\tau_{\mathcal D}(f(w,s,p)))[S^c].$$
Furthermore, recall the definition of the morphism $h_P$. For $s\in A^*$, let $s^P = h_P(s) = (s_0,P)\dots(s_{|s|-1},P)$.
Then 
$$s^P=(\tau_{\mathcal D}(f(w,s,p)))[S].$$
Hence $\tau_{\mathcal D}(f(w,s,p))$ is a shuffle of $r$ and $s^P$.
As $L$ is a commutative language, a shuffle of $r$ and $s^P$ belongs to $L$ if and only if $ r s^P \in L$.
This allows us to rewrite $\tau_{\mathcal D}(f(w,s,p)) \in L$ as $r s^P \in L$.

By the construction of $E_{u,v}$ we have for every $s,t \in E_{u,v}$ that $x s^P \in L$ iff $x {t}^P\in L$, for any $x \in (A \times \mathcal D)^*$. Hence  $\tau_{\mathcal D}(f(w,s,p)) \in L$ is independent of the choice of $s \in E_{u,v}$ as long as it has the same length as $p$, so the tuple $(w,s,p)$ remains in $\mathfrak D$.

Finally we can rewrite our set above as an intersection of four sets:
$$\{(w,s,p) \in A^*\times A^*\times (\N^2)^* \mid \exists t \in E_{u,v} \colon |t|=|p|=|s| \text{ and } r t^P \in L\} \cap \pi_c^{-1}(P) \cap \pi_2^{-1}(E_{u,v})\cap \mathfrak{D}$$
which is equal to
$$\lambda^{-1}\left(\{(w,l,p) \in A^*\times \N \times (\N^2)^* \mid \exists t \in E_{u,v} \colon |t|=|p|=l \text{ and } r t^P \in L\}\right)\cap \pi_c^{-1}(P) \cap \pi_2^{-1}(E_{u,v})\cap \mathfrak{D}.$$
By definition of $E_{u,v}$, the set $\pi_2^{-1}(E_{u,v})$ it is an element of both $\gamma_u$ and $\gamma_v$ and for all all other sets we have by the conditions on the filters that they belong to $\gamma_u$ if and only if they belong to $\gamma_v$.


\section{Proof of Completeness}\label{sec:completeness}

In this section, we will show that the families of equations we obtained in the previous section, suffice to characterize $\BlockPA$. For an identity $[u = v]$ of $\mathcal V$ denote by $\mathcal{E}_{[u=v]}$ the set of all ultrafilter equations
$[\beta f(\gamma_u) \leftrightarrow \beta f(\gamma_v)]$
where $\gamma_u$ and $\gamma_v$ satisfy the conditions of Theorem \ref{thm:main}.
We will then show that any language $L \in \mathcal{P}(A^*)$ that for each identity $[u = v]$ of $\mathcal V$ satisfies the equations of $\mathcal{E}_{u=v}$, is indeed a language in $\BlockPA$.
Since we required $\mathcal V$ to be commutative, the set $\mathcal{E}_{[ab = ba]}$ will always be among the sets of equations defining $\BlockPA$, which is an essential part for the proof of completeness.

The proof sketch follows closely the proof of completeness in \cite{GKP14}, where the completeness of equations for a fragment of logic was shown. For the more general case of an arbitrary variety $\mathcal V$ blocked with $\Parb$ a few technical changes were necessary. 

Let $w \in A^*$ and $i,j < \left| w \right|$. By $w \cdot (ij)$ denote the word that results when exchanging the letters in positions $i$ and $j$ of $w$.
Let $L$ be a language of $A^*$. Define a relation $R_L$ on $\N^2$ by
$$(i,n) R_L (j,m) \text{ iff } i+n \neq j+m \text{ or } \forall w \in A^* \quad \left|w\right| > i,j \quad \colon w \in L \Leftrightarrow w \cdot (ij) \in L$$

For a family $\mathcal{F}$ of subsets of some set $X$, we say that $\mathcal{F}$ is a \emph{filter base}, if every finite subfamily of $\mathcal{F}$ has non-empty intersection. By Zorn's Lemma, we can extend any filter base to an ultrafilter on $X$. This will be used later on to construct ultrafilters from specific filter bases, that provide us automatically with the ultrafilters having certain properties we desire.

\begin{lemma}\label{lemma:pullback}
Let $X,Y$ be two sets, $\gamma \in \beta X$ and $g: X \rightarrow Y$ a function. Then, for each $\alpha \in \beta Y$ the following conditions are equivalent
\begin{enumerate}
\item $\beta g(\gamma) = \alpha$
\item $\{g^{-1}(L) \mid L \in \alpha\} \subseteq \gamma$
\end{enumerate}
\end{lemma}

\begin{proof}
(1) implies (2). Let $L \in \alpha = \beta g(\gamma)$, then $g^{-1}(L) \in \gamma$ and thus $(2)$ holds.

(2) implies (1). Let $L \in \alpha$. Since the preimage is closed under complement, intersections and upsets and $\alpha$ is an ultrafilter, $\{g^{-1}(L) \mid L \in \alpha\}$ is an ultrafilter, too.
Since ultrafilters are maximal $\{g^{-1}(L) \mid L \in \alpha\} = \gamma$. This implies $L \in \alpha$ if and only if $g^{-1}(L) \in \gamma$ and thus $\beta g(\gamma) = \alpha$.
\end{proof}

We call the set $\{g^{-1}(L) \mid L \in \alpha\}$ \emph{pullback} of $\alpha$ by $g$, denoted by $g^{-1}(\alpha)$.

We will often use the fact that adding certain sets to a pullback of some ultrafilter will still yield a filter base.

\begin{lemma}\label{lemma:addPullback}
Let $X,Y$ be two sets, $\gamma \in \beta X$ and $g: X \rightarrow Y$ a function and $\alpha \in \beta Y$. For some set $R \subseteq X$, the condition $g(R) \in \alpha$ implies that
$g^{-1}(\alpha) \cup \{R\}$
is a filterbase.
\end{lemma}
\begin{proof}
We have to show that no element of $g^{-1}(\alpha)$ has empty intersection with $R$. Since $g(R)$ is an element of $\alpha$, we have $g(R) \cap L \neq \emptyset$ for each $L \in \alpha$. This of course implies that $g^{-1}(L) \cap R \neq \emptyset$ for all $L \in \alpha$ and thus yields the claim.
\end{proof}

\begin{theorem}\label{theorem:contPos}
Let $\gamma \in \beta( A^* \times A^* \times (\N^2)^*)$ with $k \geq 1$. Then, for each $\alpha \in \beta(\N^2)$, the following conditions are equivalent:
\begin{enumerate}
\item $\beta \contPos(\gamma) = \alpha$
\item $\{A^* \times A^* \times P^* \mid P \in \alpha\} \subseteq \gamma$
\end{enumerate}
Furthermore, these conditions hold for $\gamma$ with respect to some $\alpha$ if and only if
\begin{enumerate}
\setcounter{enumi}{2}
\item For each partition $\{P_1,\ldots,P_n\}$ of $\N^2$, we have $\bigcup_{i=1}^n (A^* \times A^* \times P_i^*) \in \gamma$
\end{enumerate}
\end{theorem}

\begin{proof}
Since $A^* \times A^* \times P^* = \contPos^{-1}(P)$, $(1)$ and $(2)$ are equivalent by Lemma \ref{lemma:pullback}.

For the second assertion, suppose there is an $\alpha \in \beta(\N)$ such that conditions $(1)$ and $(2)$ hold. Let $\{P_1,\ldots,P_n\}$ be a partition of of $\N^2$. Then $\bigcup_{i=0}^n P_i = \N^2$ and the fact that $\alpha$ is an ultrafilter implies that $P_k \in \alpha$ for some $k \in \{1,\ldots,n\}$ and thus $A^* \times A^* \times P_k^* \in \gamma$ by condition $(2)$. Since $\gamma$ is an upset, $(3)$ holds.

Suppose that $\gamma$ satsifies $(3)$ and let $\alpha = \{P \subseteq \N^2 \mid A^* \times A^* \times P^* \in \gamma\}$. Then $\alpha$ is an upset closed under intersection. For each $P \subseteq \N^2$, the partition $\{P,P^c\}$ forces either $A^* \times A^* \times P^* \in \gamma$ or $A^* \times A^* \times (P^c)^* \in \gamma$. Thus $\alpha$ is an ultrafilter and by the equivalence of $(1)$ and $(2)$ we have $\beta \contPos(\gamma) = \alpha$.
\end{proof}

\begin{lemma}\label{lemma:finIndex}
If a language $L$ of $A^*$ satisfies all the equations $\mathcal{E}_{[ab=ba]}$ for all $a,b \in A$, then $R_L$ contains an equivalence relation of finite index.
\end{lemma}

\begin{proof}
For $(a,b) \in A^2$, let
\begin{align*}
S_{ab} = \{(w,ab,(i,\left|w\right|-i)(j,\left|w\right|-j)) \in A^* \times A^* \times (\N^2)^* \mid & i < j < \left| w \right|,  w_i = a, w_j = b,\\ & w \in L \text{ but } w \cdot (ij) \notin L\}
\end{align*}
and
\begin{align*}
S_{ba} = \{(w,ba,(i,\left|w\right|-i)(j,\left|w\right|-j)) \in A^* \times A^* \times (\N^2)^* \mid & i < j < \left| w \right|,  w_i = a, w_j = b,\\ & w \in L \text{ but } w \cdot (ij) \notin L\}
\end{align*}
and
\begin{align*}
M_{ab} = \{((i,n),(j,m)) \in (\N^2)^2 \mid & \text{ there exists } w \in A^* \text{ such that } \\ & (w,ab,(i,n)(j,m)) \in S_{ab} \text{ or } (w,ab,(j,m)(i,n)) \in S_{ab}\}
\end{align*}
Then we have
$ R_L^c = \bigcup_{(a,b) \in A^2} M_{ab}.$
We show, by contraposition, that for all $(a,b) \in A^2$ there is a finite partition $\{P_1,\ldots,P_n\}$ of $\N^2$ such that the corresponding equivalence relations $\theta_{ab}$ is disjoint from $M_{ab}$. Taking the refinement of all these equivalence relation will then provide us with the desired equivalence relation contained in $R_L$. Now, suppose that for each finite partition $\{P_1,\ldots,P_n\}$ of $\N^2$,
\begin{equation}
 M_{ab} \cap \left( \bigcup_{i=1}^n P_i^2 \right) \neq \emptyset . \tag{$\ast$}
\end{equation}
Under this premise, we will construct two ultrafilters $\gamma_{ab}$ and $\gamma_{ba}$, satisfying conditions 1. - 3. of Theorem \ref{thm:main} and show that $L$ does not satisfy the equation $[\beta f(\gamma_{ab}) \leftrightarrow \beta f (\gamma_{ba})]$.
The set
$$\mathcal{F} = \{\bigcup_{i=1}^n A^* \times \N \times P_i^* \mid \{P_1,\ldots,P_n\} \text{ is a partition of } \N^2\}$$
is a filterbase on $A^* \times \N \times (\N^2)^*$. By condition $(\ast)$, $\lambda(S_{ab})$ does not have empty intersection with any of the elements of $\mathcal{F}$. Thus we can extend the filterbase $\mathcal{F}$ by $\lambda(S_{ab})$, which is equal to $\lambda(S_{ba})$, and obtain an extended filterbase. Let $\mu \in \beta(A^* \times \N \times (\N^2)^*)$ be an ultrafilter containing the extended filterbase.
Then
$$\mathcal{F}_{ab} = \lambda^{-1}(\mu) \cup A^* \times \{ab\} \times (\N^2)^*$$
is again a filterbase on $A^* \times A^* \times (\N^2)^*$. To see this, we must consider that $\mu$ was required to contain $\lambda(S_{ab})$. Since any $L \in \mu$ has non-empty intersection with $\lambda(S_{ab})$, their projection on the second component will contain $\left| ab \right|$ as an element and thus no preimage of any $L \in \mu$ will have empty intersection with $A^* \times \{ab\} \times (\N^2)^*$.
Furthermore any ultrafilter containing $\mathcal{F}_{ab}$ will also contain $S_{ab}$, since $\lambda^{-1}(\lambda(S_{ab})) \cap A^* \times \{ab\} \times (\N^2)^* = S_{ab}$.
By the same argument,
$$\mathcal{F}_{ba} = \lambda^{-1}(\mu) \cup A^* \times \{ba\} \times (\N^2)^*$$
is a filter base and any ultrafilter containing $\mathcal{F}_{ba}$ will also contain $S_{ba}$.
Let $\gamma_{ab}$ be an ultrafilter containing $\mathcal{F}_{ab}$ and respectively $\gamma_{ba}$ an ultrafilter containing $\mathcal{F}_{ba}$. Note that $S_{ba} \notin \gamma_{ab}$ and $S_{ab} \notin \gamma_{ba}$, since the two sets have empty intersection.
By Lemma \ref{lemma:pullback} and Theorem \ref{theorem:contPos} the two ultrafilters satisfy
$$\beta \contPos(\gamma_{ab}) = \beta \contPos(\gamma_{ba}) \in \beta(\N^2) \text{ and } \beta \lambda(\gamma_{ab}) = \beta \lambda(\gamma_{ba}).$$
Since ultrafilters are upsets and $\gamma_{ab}$ contains $A^* \times \{ab\} \times (\N^2)^*$, we have that $ab$ as an element of $\mathcal(\Reg)$ is a subset of $\beta \pSub (\gamma_{ab})$ and respectively $ba \subseteq \beta \pSub (\gamma_{ba})$.
By definition $f(S_{ab}) \subseteq L$ or equivalently $S_{ab} \subseteq f^{-1}(L)$ and $S_{ba} \subseteq f^{-1}(L^c)$. Thus $L \in \beta f (\gamma_{ab})$ but $L \notin \beta f (\gamma_{ba})$.
By contraposition, if $L$ satisfies $\mathcal{E}_{[ab = ba]}$, then there is an equivalence relation $\theta_{ab}$ of finite index which is disjoint for $M_{ab}$. Setting $\theta = \bigcap_{a,b \in A} \theta_{ab}$, we see that $\theta$ is an equivalence relation of finite index contained in $R_L$ since
$$ \theta = \bigcap_{a,b \in A} \theta_{ab} \subseteq \bigcap_{a,b \in A} M_{ab}^c = R_L $$
\end{proof}

A direct consequence is the following Corollary, which makes use of the fact that each finite equivalence class can be split into singletons and still yields an equivalence relation.

\begin{corollary}\label{corollary:finIndex}
If a language $L$ of $A^*$ satisfies the equations $\mathcal{E}_{[ab=ba]}$ for all $a,b \in A$, then $R_L$ contains an equivalence relation of finite index for which each finite equivalence class is a singleton.
\end{corollary}

We use the Eilenberg correspondence between varieties of regular languages and varieties of finite monoids. By $\mathbf V$ denote the variety of finite monoids associated with $\mathcal V$.

By $\MorphV$ denote the set of all morphisms from $A^*$ into a monoid of $\VMon$. $\MorphV$ is countable, since all monoids are finite and hence there are countably many morphisms into monoids of $\VMon$.
Then there exists a bijection $\phi_{\text{Hom}} : \N \rightarrow \MorphV$. As a shorthand define $h_i := \phi_{\text{Hom}}(i)$, where $h_i: A^* \rightarrow M_i$ and $M_i \in \VMon$.

The space $\N \times \MorphV$ thus is countable, too and hence any family of words of $A^*$ indexed by $\N \times \MorphV$ is a sequence. Let $\phi: \N \times \MorphV \rightarrow \N$ be a bijection and $(s_n)_{n \in \N}$ be a sequence of words. For $n \in \N$ and $h \in \MorphV$, define $s(n,h) := s_{\phi(n,h)}$.

\begin{lemma}
Let $(s_n)_{n \in \N}$ and $(t_n)_{n \in \N}$ be two sequences of words of $A^*$ satisfying the property $h(s(n,h)) = h(t(n,h))$ for all $h \in \MorphV$.
Then for each $N \in \N$ there exists a fixed morphism $\varphi_N \in \MorphV$, such that for all $i \leq N$:
$h_i (s(N,\varphi_N)) = h_i (t(N,\varphi_N))$
\end{lemma}

\begin{proof}
Let $N \in \N$. By $M_i$ denote the monoid that $h_i$ maps into. Define $\varphi_N : A^* \rightarrow M_0 \times \ldots \times M_{N-1}$ by  $\varphi_N (w) = (h_0(w), \ldots, h_{N-1}(w))$. This makes $\varphi_N$ a morphism of $\MorphV$, since $\VMon$ is closed under finite products.
The condition $\varphi_N(s(N,\varphi_N)) = \varphi_N(t(N,\varphi_N))$ then implies $h_i (s(N,\varphi_N)) = h_i (t(N,\varphi_N))$.
\end{proof}

Note that we may choose a bijection $\phi$, that satisfies the property that for $n \leq m$, we have $\phi(n,h) \leq \phi(m,h)$ for all morphisms $h \in \MorphV$. This allows for the following Corollary. 

\begin{corollary}\label{corollary:subSeq}
Let $(s_n)_{n \in \N}$ and $(t_n)_{n \in \N}$ be two sequences of words of $A^*$ satisfying the property $h(s(n,h)) = h(t(n,h))$ for all $h \in \MorphV$.
Then there exist subsequences $(s_{m_n})_{n \in \N}$ and $(t_{m_n})_{n \in \N}$ such that for all $i \leq m_n$:
$h_i(s_{m_n}) = h_i(t_{m_n}).$
\end{corollary}

\begin{proof}
For $N \in \N$ define $\varphi_N$ as in the previous lemma. Set $m_N := \phi(N,\varphi_N)$.
\end{proof}

For a word $w \in A^*$ and $P \subseteq \N$ with $\{(0,\left|w\right|-1),(1,\left|w\right|-2),\ldots,(\left|w\right|-1,0)\} \cap P = \{p_1,\ldots,p_k\}$ where $p^1_1 < p^1_2 < \ldots < p^1_k$ define
$$w[P] = w_{p_1^1}\ldots w_{p_k^1}$$
and
$$ P[w] = p_1\ldots p_k \in (\N^2)^*.$$

\begin{lemma}\label{lemma:exPhn}
Let $L$ be a language of $A$ satisfying all the equations $\mathcal{E}_{[u = v]}$. Let $\theta$ be an equivalence relation of finite index contained in $R_L$ and let $P$ be an infinite equivalence class of $\theta$. Then there exists an $n \in \N$ and a morphism $h: A^* \rightarrow M$ into a finite monoid $M \in \VMon$
such that for all $s,t \in A^*$, if
\begin{enumerate}
\item $n \leq \left| s \right| = \left| t \right|$,
\item $s_i = t_i$ for all $i \notin P$
\item $h(s[P]) = h(t[P])$
\end{enumerate}
then
$ s \in L \Leftrightarrow t \in L$.
\end{lemma}

\begin{proof}
By contraposition. Suppose that for every $n \in \N$ and morphism $h: A^* \rightarrow M$ into a finite monoid there exist two words $s(n,h)$ and $t(n,h)$ such that $(1)-(3)$ hold and $s(n,h) \in L$, but $t(n,h) \notin L$. 

Recall that $\phi: \N \times \MorphV \rightarrow \N$ is a bijection. Let $s_n = s(\phi^{-1}(n))$ and $t_n = t(\phi^{-1}(n))$. Considering the sequences $(s_n[P])_{n \in \N}$ and $(t_n[P])_{n \in \N}$, condition $(3)$ provides us with $h(s(n,h)[P]) = h(t(n,h)[P])$.

Then by Corollary $\ref{corollary:subSeq}$ there exist subsequences $(s_{m_n}[P])_{n \in \N}$ and $(t_{m_n}[P])_{n \in \N}$ such that for all $i \leq m_n$:
$h_i(s_{m_n}[P]) = h_i(t_{m_n}[P]).$

As $A^*$ can be embedded into $\Ahat$, both $(s_{m_n}[P])_{n \in \N}$ and $(t_{m_n}[P])_{n \in \N}$ define sequences in the free profinite monoid. Since this space is compact, every sequence has a convergent subsequence, hence there exists a set $J \subseteq \{m_n \mid n \in \N\}$ such that $(s_j[P])_{j \in J}$ converges and a set $I \subseteq J$ such that both $(s_i[P])_{i \in I}$ and $(t_i[P])_{i \in I}$ converge.

Define $u := \lim_{i \in I} s_i[P]$ and $v := \lim_{i \in I} t_i[P]$. We claim that $u \equiv_{\mathbf{V}} v$, so $\hat{h}(u) = \hat{h}(v)$ for every morphism $h: A^* \rightarrow M$ into a monoid $M \in \VMon$, where $\hat{h}$ denotes its unique continuous extension to $\Ahat$.

Let $h \in \MorphV$, then there exists an $i_0 \in \N$ such that $h = h_{i_0}$. Hence for all $i > i_0$ we have $h(s_i[P]) = h(t_i[P])$, which implies $\hat{h}(u) = \hat{h}(v)$, since $\hat{h}$ is continuous.
Define
$ T_s = \{(s_n,s_n[P],P[s_n]) \mid n \in \N\} $
and
$ T_t = \{(s_n,t_n[P],P[s_n]) \mid n \in \N\} $
By (2), for $n \in \N$ we obtain $f(s_n,s_n[P],P[s_n]) = s_n \in L$ and $f(s_n,t_n[P],P[s_n]) = t_n \notin L$
and thus
$f(T_s) \subseteq L \text{ and } f(T_t) \subseteq L^c.$
We claim that there exist two ultrafilters $\gamma_u$ and $\gamma_v$ satisfying
\begin{enumerate}
\item $\beta\lambda(\gamma_u) = \beta\lambda(\gamma_v)$
\item $u \subseteq \beta \pSub(\gamma_u) \text{ and } v \subseteq \beta \pSub(\gamma_v)$
\item $\beta \contPos(\gamma_u) = \beta \contPos(\gamma_v) \in \beta(\N^2)$
\end{enumerate}
such that
$L \in \beta f(\gamma_u)$ and $L^c \in \beta f(\gamma_v)$.
In order to ensure that all three conditions hold, we will use the technique already applied in Lemma \ref{lemma:finIndex}. That is by starting with one ultrafilter and using pullback while subsequently adding sets to the resulting filterbases, that yield the desired properties.

To ensure property 3., let $\alpha \in \beta(\N^2)$ with $P \in \alpha$. As a reminder, we denoted both the projections from $A^* \times A^* \times (\N^2)^*$ and $A^* \times \N \times (\N^2)^*$ onto the content of the third component by $\contPos$. Thus the pullback of $\alpha$ by $\contPos$ provides us with a filterbase on $A^* \times \N \times (\N^2)^*$.
\[
A^* \times \N \times (\N^2)^* \xhookleftarrow{\contPos^{-1}(\alpha)} \N^2
\]
Furthermore we have that, $\lambda(T_s) = \lambda(T_t)$, which implies that $\contPos(T_s) = \contPos(T_t) = P$. Thus, adding the set $\{\lambda(T_s)\}$ to $\contPos^{-1}(\alpha)$ still yields a filterbase by Lemma \ref{lemma:addPullback}. Let $\mu \in \beta(A^* \times \N \times (\N^2)^*)$ containing the extended filter base. Recall that the mapping $\contPos$ factors through $\lambda$ by $\contPos \circ \lambda = \contPos$ and thus $\contPos^{-1}(\alpha) \subseteq \lambda^{-1}(\mu)$. This ensures that any ultrafilter $\gamma$ containing $\lambda^{-1}(\mu)$ will satisfy $\beta\lambda(\gamma) = \mu$ and $\beta\contPos(\gamma) = \alpha$.
\[
A^* \times A^* \times (\N^2)^* \xhookleftarrow{\lambda^{-1}(\mu)} A^* \times \N \times (\N^2)^* \xhookleftarrow{\contPos^{-1}(\alpha)} \N^2
\]
Since $\lambda(T_s) \in \mu$, the sets
$\mathcal{F}_u = \lambda^{-1}(\mu) \cup \pSub^{-1}(u) \cup \{T_s\}$
and
$\mathcal{F}_v = \lambda^{-1}(\mu) \cup \pSub^{-1}(v) \cup \{T_t\}$
are both filterbases.

Any ultrafilter $\gamma_u$ containing $\mathcal{F}_u$ and $\gamma_v$ containing $\mathcal{F}_v$ will satisfy $1. - 3.$.

Let $\gamma_u$ and $\gamma_v$ be two such ultrafilters. Then $T_s \in \gamma_u$ and since $f(T_s) \subseteq L$, we obtain $T_s \subseteq f^{-1}(L)$ and thus $f^{-1}(L) \in \gamma_u$ and by $T_t \subseteq L^c$, $f^{-1}(L^c) \in \gamma_v$. Thus
$L \in \beta f (\gamma_u) \text{ and } L \notin \beta f(\gamma_v).$
By contraposition, the claim holds.
\end{proof}

\begin{lemma}\label{lemma:exhn}
Let $L$ be a language of $A^*$ satisfying all the equations $\mathcal{E}_{[u=v]}$ and let $\theta$ be an equivalence class of finite index contained in $R_L$. Then there exists an $n \in \N$ and a morphism $h: A^* \rightarrow M$ into a finite monoid $M \in \VMon$ such that for all $s,t \in A^*$, if $n \leq \left| s \right| = \left| t \right|$ and
\[ h(s[P]) = h(t[P]) \text{ for each $\theta$ equivalence class $P$,}\]
then
$s \in L \Leftrightarrow t \in L$.
\end{lemma}
\begin{proof}
Let $L$ satisfy the equations $\mathcal{E}_{[ab = ba]}$. Then, by Corollary \ref{corollary:finIndex}, $R_L$ contains an equivalence relation of finite index $\theta$ for which each finite equivalence class is a singleton. Let $P_1,\ldots,P_r$ be the equivalence classes of $\theta$. For each $i \in \{1,\ldots,r\}$ with $P_i$ infinite, we define $n_i$ and $h_i$ as in Lemma \ref{lemma:exPhn}. Furthermore define
$ n = \max\{n_i \mid P_i \text{ is infinite }\}$
and
$ h(u) = (h_1(u),\ldots,h_r(u)).$
Again, $h$ is a morphism into a monoid of $\VMon$, since $\VMon$ is closed under finite products.
Now let $u,v \in A^*$, with $n \leq \left| u \right| = \left| v \right|$ and $h(u[P]) = h(v[P])$ for each $\theta$ equivalence class $P$. We define words $w_i \in A^*$ for $i = 0,\ldots,n$ and $j = 0,\ldots, \left|u \right|$ by
$$
(w_i)_j = 
\begin{cases}
u_j & \text{if $j \in P_k$ and $i<k$}\\
v_j & \text{otherwise.}
\end{cases}
$$

By construction we have $w_0 = u$, $w_n = v$ and Lemma \ref{lemma:exPhn} applies to each pair $w_{i-1},w_i$ with $i \in \{1,\ldots,n\}$ and thus
$ w_{i-1} \in L \Leftrightarrow w_i \in L .$
It follows that
$ u \in L \Leftrightarrow v \in L .$
\end{proof}

For $N \in \N$ denote by $A^{\geq N}$ the set of all words of length greater or equal to $N$, that is
$A^{\geq N} = \{w \in A^* \mid \left|w \right| \geq N\}$

\begin{theorem}
If $L \in \mathcal{P}(A^*)$ satisfies all the equations $\mathcal{E}_{[u = v]}$, then $L \in \BlockPA$.
\end{theorem}

\begin{proof}
Let $h \colon A^* \rightarrow M$ be a morphism into a monoid of $\VMon$. For $P \subseteq \N^2$ and $m \in M$ define the set
$ L_{P,m} = \{w \in A^* \mid h(w[P]) = m\}.$
Since $M \in \VMon$, the language $R = h^{-1}(m)$ is an element of $\mathcal V_{A}$. Let $\mathcal D = \{P,P^c\}$ be a partition of $\N^2$. Define the morphism
\begin{align*}
e_{\mathcal D} \colon (A \times \mathcal D)^* &\rightarrow A^*\\
(w,P) & \mapsto w\\
(w,P^c) & \mapsto \epsilon
\end{align*}
Since $\mathcal V_{A}$ is closed under inverse morphisms, $e^{-1}_{\mathcal D}(h^{-1}(m)) \in \mathcal V_{A \times \mathcal D}$.
Then
$$L_{P,m} = \{w \in A^* \mid \tau_{\mathcal D} \in e^{-1}_{\mathcal D}(h^{-1}(m))\}$$
is an element of $\BlockPA$.
By Corollary \ref{corollary:finIndex}, the relation $R_L$ contains an equivalence relation $\theta$ of finite index for which each finite equivalence class is a singleton. Let $P_1,\ldots,P_r$ be the corresponding partition of $\N$. By Lemma \ref{lemma:exhn}, there exists an $N \in \N$ and a morphism $h \colon A^* \rightarrow M$ into a monoid $M \in \VMon$ such that for $m \in M$ and
$$L_m = \left( \bigcap_{i=1}^r L_{P_i,m} \right) \cap A^{\geq N}$$
either $L_m \subseteq L$ or $L_m \subseteq L^c$.

This implies that there exists some $Q \subseteq M$ such that
$L \cap A^{\geq N} = \bigcup_{m \in Q} L_m.$
Since $\BlockPA$ contains all finite languages, $L$ is a Boolean combination of languages in $\BlockPA$.
\end{proof}

\section{Conclusion}\label{sec:conclusion}

We have presented a method applicable to arbitrary classes of languages, to describe circuit classes by equations. The tools and techniques used originate from algebra and topology and have previously been used on regular language classes. Due to recent developments in generalizing these methods to non-regular classes, they are now powerful enough to describe circuit classes. But the knowledge that they are powerful enough itself is not sufficient, as we require a constructive mechanism behind these descriptions. Since non-uniform circuit classes are by definition not finitely presentable, this seemed to be impossible.

Nevertheless, we were able to find a description of small but natural circuit classes via equations. This description seems helpful as it easily allows to prove non-membership of a language to some circuit class. Another advantage is the possibility of using Zorn's Lemma for the extension of filter bases to ultrafilters, which prevents us from having to use probabilistic arguments in many places. Also in Lemma \ref{lemma:exPhn} we use purely topological arguments of convergence, for which it is unclear how this could be achieved purely combinatorially.

The results we acquired are not so different from the results about equations for varieties of regular languages by Almeida and Weil \cite{AlWe98}. This gives hope that their results can be used as a roadmap for further research. 

In \cite{WBP1} it was shown that a certain restricted version of the block product of our constant size circuit classes would actually yield linear size circuit classes (over the same base). Here having equations for all languages captured by this circuit class, not just the regular ones, would pay off greatly. By showing that a padded version of a language is not in a linear circuit class  we could already prove that PARITY is not in a polynomial size circuit class.
Equations for non-regular language classes could be used to overcome previous bounds. The separation results in the corollary can easily be extended to show that a padded version of those languages is not contained in these circuit classes.

A different approach would be to examine the way the block product was used here. The evaluation of a circuit is equivalent to a program over finite monoids. While the program itself has little computational power, it allows non-uniform operations like our $\N$-transducers. The finite monoid itself corresponds loosely speaking to the computational power of the gates of the circuit, which was handled by our variety $\mathcal V$. For general circuit classes one would need to consider larger varieties containing also non-commutative monoids. While the methods here seem to be extendable to non-commutative varieties, the more complicating problem remaining is to find an extension of the block product that corresponds to polynomial programs over these monoids.

\bibliography{literature}
\bibliographystyle{plain}

\end{document}